\documentclass[letterpaper]{article}

\usepackage[margin=1.2in]{geometry}
\usepackage[utf8]{inputenc}
\usepackage[T1]{fontenc}
\usepackage{lmodern}
\usepackage{microtype}
\usepackage{pifont}
\usepackage[toc,page]{appendix}
\usepackage[ruled,vlined, linesnumbered]{algorithm2e}

\makeatletter
\renewcommand{\nllabel}[1]
 {{\let\@currentlabel\algocf@currentlabel
  \let\@currentcounter\algocf@currentcounter
  \label{#1}}}

\renewcommand{\algocf@nl@sethref}[1]{\renewcommand{\theHAlgoLine}{\thealgocfproc.#1}\hyper@refstepcounter{AlgoLine}\gdef\algocf@currentlabel{#1}\gdef\algocf@currentcounter{AlgoLine}}\makeatother

\usepackage[english]{babel}

\usepackage[l2tabu, orthodox, experimental]{nag}

\usepackage{hyperref}
\makeatletter
\hypersetup{
	unicode    = true,
	pdftitle   = "\@title",
	pdfauthor  = "\@author"
}
\makeatother

\usepackage{csquotes}

\usepackage{bm}

\usepackage{enumitem}
\setenumerate[1]{label=\textnormal{(\alph*)}}
\setlist[enumerate,2]{label=\textnormal{(\roman*)}, ref=\theenumi.(\roman*)}

\usepackage{xspace}

\usepackage{makeidx}

\usepackage[dvipsnames]{xcolor}

\makeatletter
\g@addto@macro\bfseries{\boldmath}
\makeatother

\usepackage{todonotes}

\usepackage{amsmath, amsfonts, amssymb, amsthm, mathtools, nicefrac, cancel,thmtools}

\usepackage[capitalize, nameinlink]{cleveref}

\newcommand*{\dotcup}{\mathbin{\dot\cup}}

\newcommand{\symdiff}{\mathrel{\triangle}}

\newcommand{\N}{\mathbb{N}}

\DeclareMathOperator{\image}{im}

\DeclareMathOperator{\Aut}{Aut}

\DeclareMathOperator{\Iso}{Iso}
\DeclareMathOperator{\Sym}{Sym}
\DeclareMathOperator{\Stab}{Stab}
\DeclareMathOperator{\Trans}{Trans}

\DeclareMathOperator{\VCdim}{VC-dim}
\DeclareMathOperator{\pVCdim}{pVC-dim}

\newcommand*{\simeqq}{\cong}

\theoremstyle{definition}
\newtheorem{definition}{Definition}[section]

\theoremstyle{plain}
\newtheorem{lemma}[definition]{Lemma}

\newtheorem{theorem}[definition]{Theorem}	
\newtheorem{corollary}[definition]{Corollary}

\newtheorem{claim}{Claim}[definition]

\Crefname{claim}{Claim}{Claims}
\Crefname{observation}{Observation}{Observations}

\newenvironment{claimproof}{\proof[\(\ulcorner\)\nopunct]}{\endproof}

\newcommand{\ceq}{\succeq}

\makeatletter
\let\@vareps\varepsilon
\let\varepsilon\epsilon
\let\epsilon\@vareps
\makeatother

\renewcommand{\phi}{\varphi}

\usepackage{tikz, pgf}
\usetikzlibrary{calc, shapes,arrows.meta,shapes.geometric,positioning,fit, backgrounds, decorations.pathmorphing, decorations.pathreplacing}

\tikzstyle{vertex}=[circle,draw=black,fill=black, minimum size=2mm, inner sep=0mm]

\title{Isomorphism of tournaments with bounded VC dimension}
\author{Simon Raßmann \hspace{1cm} Pascal Schweitzer\\[4mm]TU Darmstadt}
\date{}
\makeatletter
\hypersetup{pdftitle="\@title"}
\makeatother

\definecolor{lightblue}{rgb}{0.5,0.5,1.0}
\definecolor{darkred}{rgb}{0.5,0,0}
\definecolor{darkgreen}{rgb}{0,0.5,0}
\definecolor{darkblue}{rgb}{0,0,0.5}

\hypersetup{colorlinks,linkcolor=darkred,filecolor=darkgreen,urlcolor=darkred,citecolor=darkblue}

\begin{document}

\maketitle
\begin{abstract}
	The tournament isomorphism problem is one of the two fundamental bottlenecks to designing better algorithms for the graph isomorphism problem. Though the problem has been investigated for more than five decades, compared to graphs, there are only very few results on the isomorphism problem of tournaments. For most classes of tournaments neither hardness nor polynomial-time solvability is known.
	
	Tournaments of bounded VC dimension are such a class for which no results are available,
	even though the VC dimension is arguably one of the most robust and central notions of combinatorial tameness. 
	
	Resolving an open problem of Neuen and Grohe, we show that the isomorphism problem for tournaments of VC dimension \(d\) can be decided in time \(n^{O(d\log d)}\).
	Consequently, automorphism groups of tournaments of bounded VC dimension can be computed in polynomial time.
	
	To this end, we develop a new method to isomorphism-invariantly decompose tournaments.
	To facilitate recursion, we introduce the notion of a patched tournament and analyze bounded VC dimension in patched tournaments.
	We design a recursive algorithm that balances the size of the decomposed pieces against their number and makes use of the structure of near twins.
	
	In an orthogonal direction, it is known that a hereditary class of tournaments has unbounded VC dimension if and only if it contains
	all \(2\)-colorable tournaments.
	As a second result, we show that also this class does not form an obstruction towards polynomial-time isomorphism testing
	and indeed show that isomorphism of tournaments of bounded chromatic number is polynomial-time decidable.
\end{abstract}

\paragraph{Statement of AI use.}
The proof of \Cref{lem:prelims:tournament_VCdim} was provided by Claude Sonnet 5,
and the proof of \Cref{lem:patched:VC-dim_bound} was simplified by Claude Sonnet 5.

\section{Introduction}
With the infamous graph isomorphism problem having a quasipolynomial-time algorithm due to Babai's 2015 breakthrough~\cite{DBLP:conf/stoc/Babai16}, two special cases are now regarded as prime obstacles for further improvements. Specifically, these are the group isomorphism problem (for groups given verbosely by multiplication table) and the tournament isomorphism problem. Both problems are concerned with basic structures whose isomorphism problems naturally reduce to the graph isomorphism problem. Both of these problems admit quasipolynomial-time algorithms, known since the 1980s~\cite{DBLP:conf/stoc/BabaiL83, group_isomorphism}. In particular, these algorithms are comparatively simple and do not require any of the new machinery developed by Babai just over a decade ago for graphs.

As special cases, the two isomorphism problems focus our view on particular barriers in isomorphism testing for which we have no efficient techniques yet. While groups come with a stringent algebraic structure, which we do not know how to exploit, tournaments come with the promise that their automorphism group is solvable, which we also do not know how to exploit. This is frustrating, since the group-theoretic machinery developed in the 1980s by Babai, Luks, and others, works particularly well for solvable groups. Indeed, many problems from the area of computational group theory (e.g.~set stabilizer, string isomorphism, centralizer) can be solved in polynomial time for solvable 
groups. This implies that automorphism or isomorphism problems can also be solved in polynomial time whenever the search space of isomorphisms lies inside a given ``encasing group'' (see~\cite{DBLP:conf/stoc/Babai16}) that is solvable.
In some sense, the ``only'' issue regarding tournament isomorphism is therefore to find such an encasing group. Such a group is known to exist, since the sought-after automorphism groups always are themselves examples. In this sense the tournament isomorphism problem singles out the group-theoretic difficulty of demonstrating that the automorphism group is not complex. In fact, for randomized algorithms it even suffices to decide whether the automorphism group is trivial~\cite{DBLP:conf/icalp/Schweitzer17}.

Nevertheless, in general, the canonical labeling algorithm from 1983 by Babai and Luks~\cite{DBLP:conf/stoc/BabaiL83}, which runs in~$n^{O(\log n)}$ time, is still the best known algorithm for the tournament isomorphism problem. Regarding specialized classes, very little is known (see related work below).
Most recently, an algorithm for tournaments of bounded twin-width~\cite{tournament_isomorphism_tww} was designed. In fact that algorithm is even fixed-parameter tractable.
For tournaments, bounded twin-width classes turn out to be exactly those classes of tournaments that are monadically dependent (NIP)~\cite{tww_tournaments}, an important notion from classical model theory. The question presents itself where exactly the boundary of tractability for the isomorphism problem of tournaments should lie. 
Towards this, Neuen and Grohe~\cite{tournament_isomorphism_tww} ask whether isomorphism of tournaments of bounded VC dimension can be solved in polynomial time. Indeed, bounded twin-width implies bounded VC dimension, but not vice versa.

The VC dimension of a family of subsets~\(\mathcal{F}\subseteq\mathcal{P}(X)\) of some set~$X$ essentially measures the richness of the set family. 
Roughly speaking it measures the richness of the intersection patterns \(\{F\cap Y\colon F\in\mathcal{F}\}\) for subsets~$Y\subseteq X$. For the VC dimension of a directed graph it is at first unclear which set family to take, but it turns out that the natural choices yield functionally equivalent definitions. One such choice is to take as family~$\mathcal{F}$ the set of out-neighborhoods of the vertices.

In fact,  the VC dimension is arguably one of the most robust and central notions of combinatorial tameness.
It  appears in learning theory, set theory, finite model theory (in particular for model checking), incidence geometry, and computational geometry.  For hereditary classes, a striking combinatorial characterization says that every class with growth less than~\(2^{n^2/4}\) (rather than~\(2^{\frac{n^2-n}{2}}\) like the class of all tournaments) has bounded VC dimension~\cite{domination_tournaments}. 
This is equivalent to excluding a 2-colorable tournament, where a tournament is~$k$-colorable if there exists a partition \(V(T)=V_1\dotcup\dots\dotcup V_k\) such that each part induces a transitive (i.e., acyclic) tournament. 
Overall, it is therefore natural to ask whether bounding the VC dimension alone provides enough structure to make tournament isomorphism tractable.

\textbf{Our result.} In this paper we prove that isomorphism for tournaments of VC dimension \(d\) can be tested in polynomial time \(n^{O(d\log d)}\).
Consequently, we can also compute the automorphism group of a tournament of bounded VC dimension in polynomial time.

Complementing the result, we show that also isomorphism of~$k$-colorable tournaments can be decided in polynomial time \(n^{O(k)}\).

\textbf{Techniques.} 
The starting point of our approach is the concept of near twins, that is, vertices whose out-neighborhoods are fairly similar. 
It turns out, as essentially observed in~\cite{tournament_isomorphism_tww}, that in a tournament there is a favorable upper bound on the number of near twins a vertex can have (see Lemma~\ref{lem:near-twins:out-degree}). On the other hand, in a tournament of restricted VC dimension, there is also a favorable upper bound on the size of a set of vertices that are pairwise not near-twins (see \Cref{cor:near-twins:existence}).
Thus, if we consider the graph induced by the relation of being near twins, we obtain a graph in which we have a bound on the degree and also a bound on the number of connected components. 

Conceptually, we might therefore be able to perform isomorphism tests for the connected components using bounded degree techniques that have been developed within the group-theoretic machinery for graph isomorphism testing.
To understand how the connected components are permuted as a collection, we would use a quotient tournament obtained by contracting each connected component to a single vertex and contracted edges using the majority direction. This trick of forming the quotient tournament requires that the components have odd order, which can typically be achieved using other group-theoretic techniques (degree-refinement and orbit-by-orbit processing). We then would recurse on the components and on the quotient tournament. Since automorphism groups of tournaments are solvable, we obtain a solvable encasing group (a direct-or-wreath product of solvable groups).
Assembling the results can then be achieved using algorithms for the set stabilizer problem for solvable groups~\cite{DBLP:conf/dimacs/Luks91}.

The central difficulty with this idea is that for recursive calls of our algorithm, we will need again a good bound on the VC dimension. However, we are not aware of an argument with which the VC dimension of the quotient tournament can be controlled.
To remedy this, we found a crucial combinatorial lemma on the VC dimension for a type of structure which we call a patched tournament. In fact, if we pick a partition of the vertices in a tournament and remove all arcs whose endpoints are in the same part, we obtain an oriented graph whose VC dimension can be bounded by the VC dimension of the tournament (see~\Cref{lem:patched:VC-dim_bound}).
The downside of this is now that all the recursive arguments have to work with patched tournaments. We need to replace the notion of a component of near-twins by what we call chunks, and we need to adapt the group-theoretic machinery to patched tournaments. Unfortunately, this makes several of the arguments somewhat technical, as we need to ensure solvability of the involved objects. For this, we move to solvable groupoids rather than solvable groups.
To assemble the results in the end, we then also make use of the fact that hypergraph isomorphism for solvable encasing groups is polynomial-time solvable~\cite{hypergraph_isomorphism}.

To show that isomorphism of~$k$-colorable tournaments is decidable in polynomial time we design a recursive algorithm.
To achieve this recursion, we argue that in sufficiently regular tournaments, we can always find a vertex such that both its in- and its out-neighborhood
induce subtournaments of smaller chromatic number. The crucial obstacle to overcome here is that testing~$k$-colorability is \textsc{NP}-complete.
Instead, we use the success of our algorithm itself as a polynomial-time decidable approximation to \(k\)-colorability of the neighborhoods.
This allows us to either find a nontrivial isomorphism-invariant vertex coloring, which allows us to recurse on the color classes,
or have a sufficiently regular tournament such that we can recurse on the in- and out-neighborhood of some vertex. See Section~\ref{sec:k:colorable:tournaments}.

\textbf{Further related work.} In 1992,  Ponomarenko designed a polynomial-time isomorphism algorithm for circulant tournaments~\cite{Ponomarenko1992}. This was later generalized to graphs~\cite{DBLP:journals/endm/EvdokimovP05} and relational structures~\cite{DBLP:journals/jct/MuzychukP20}. In the realm of edge colored tournaments, it is known that Schurian tournament isomorphism can be solved in polynomial time (\cite{MR2981982}), a class defined algebraically through coherent configurations. 
Techniques from a recent polynomial-time (FPT) isomorphism algorithm for edge colored tournaments for which one color class is spanning and has bounded degree~\cite{DBLP:journals/adam/ArvindPR25} found their way into the result on twin-width mentioned above. It appears that these are the only classes for which polynomial-time isomorphism algorithms have been developed. 

Regarding lower bounds, it is known that the problem is hard for~$\textsc{NL}$,~$\textsc{C=L}$, and~$\textsc{PL}$~\cite{DBLP:conf/mfcs/Wagner07}. There is also a (peculiarly randomized) reduction from tournament isomorphism to tournament asymmetry~\cite{DBLP:conf/icalp/Schweitzer17} and a similar statement in the context of canonization~\cite{DBLP:conf/isaac/ArvindDM06}.

A robust version of the graph isomorphism problem was considered in~\cite{DBLP:journals/corr/abs-2604-12584} and shown to be solvable when the VC dimension is bounded. This problem is orthogonal to the isomorphism problem considered here and calls for very different techniques and solutions.

Finally, we want to highlight that for general graphs, isomorphism of bounded VC dimension is GI-complete, as can be seen by using subdivisions.

 \section{Preliminaries}\label{sec:prelims}

\paragraph{General Notation.} We denote by~$\mathcal{P}(X)$ the \emph{power set} of a set~$X$.
The symmetric difference of two sets~$X$ and~$Y$ is denoted~$X\symdiff Y$.
If \(\mathcal{P}\) and \(\mathcal{Q}\) are partitions of some set \(X\), we say \(\mathcal{P}\) is \emph{coarser} than \(\mathcal{Q}\),
and write \(\mathcal{P}\ceq\mathcal{Q}\), if every part of \(\mathcal{P}\) is a union of parts of \(\mathcal{Q}\).

A (finite) digraph is a pair \(D=(V(D),E(D))\) where \(V(D)\) is a finite set of \emph{vertices} and the set of arcs \(E(D)\subseteq V(D)^2\)
is an irreflexive binary relation. We denote an arc from \(v\) to \(w\) as \(vw\).
An out-neighbor of a vertex \(v\in V(D)\) is a vertex \(w\in V(D)\) such that \(vw\in E(D)\).
The out-neighborhood \(N^+_D(v)\) is the set of out-neighbors of \(v\). Similarly,
the in-neighborhood \(N^-_D(v)\) is the set of vertices that have \(v\) as an out-neighbor.
The out-degree of \(v\) is \(\deg^+_D(v)\coloneqq|N^+_D(v)|\) and similarly the in-degree is \(\deg^-_D(v)\coloneqq|N^-_D(v)|\).

A digraph is \emph{oriented} if there are no two vertices \(v,w\in V(D)\) such that both \(vw\in E(D)\) and \(wv\in E(D)\).
A tournament is a digraph such that for any two distinct vertices \(v,w\in V(D)\) exactly one of the two pairs \(vw\) and \(wv\) is an arc.

\paragraph*{VC dimension.}
Let \(X\) be a finite set and \(\mathcal{F}\subseteq\mathcal{P}(X)\) a family of subsets. A set \(Y\subseteq X\) is \emph{shattered}
if \(\{F\cap Y\colon F\in\mathcal{F}\}=\mathcal{P}(Y)\). The size of a largest set that is shattered by \(\mathcal{F}\) is called
the \emph{VC dimension} of \(\mathcal{F}\) and denoted by \(\VCdim(\mathcal{F})\).

One of the most important lemmas in the theory of the VC dimension is the Sauer-Shelah Lemma:
\begin{lemma}[{\cite{sauer_shelah1,sauer_shelah2,VC}}]\label{lem:prelims:sauer-shelah}
Let \(X\) be a finite set, \(\mathcal{F}\subseteq \mathcal{P}(X)\)
a family of subsets of \(X\) and \(Y\subseteq X\). Then
\[\left|\{F\cap Y\colon F\in\mathcal{F}\}\right|\leq\sum_{i=0}^{\VCdim(\mathcal{F})}\binom{|Y|}{i}\leq 1+|Y|^{\VCdim(\mathcal{F})}.\]
\end{lemma}

For a set family \(\mathcal{F}\subseteq\mathcal{P}(X)\), we can also define the dual set family
\[\hat{\mathcal{F}}\coloneqq\{\{F\in\mathcal{F}\colon x\in F\}\colon x\in X\}\subseteq\mathcal{P}(\mathcal{F})\]
as a family of subsets of \(\mathcal{F}\). It is routine to check that the dual of \(\hat{\mathcal{F}}\) is again \(\mathcal{F}\),
and by \cite[Proposition 2.12]{VCdim_dual} it holds that
\[\VCdim(\mathcal{F})<2^{\VCdim(\hat{\mathcal{F}})+1}.\]

Given a digraph \(D\), consider the set families \(\mathcal{N}^+(D)\coloneqq\{N^+(v)\colon v\in V(D)\}\) and \(\mathcal{N}^-(D)\coloneqq\{N^-(v)\colon v\in V(D)\}\) as families of subsets of \(V(D)\), and \(\VCdim^+(D)=\VCdim(\mathcal{N}^+(D))\) and \(\VCdim^-(D)=\VCdim(\mathcal{N}^-(D))\).
Note that these two set families are duals of one another, which implies
\[\VCdim^+(D)<2^{\VCdim^-(D)+1}\qquad\text{and}\qquad\VCdim^-(D)<2^{\VCdim^+(D)+1}.\]
In particular, one of these parameters is bounded on a class of digraphs if and only if the other one is bounded.
If either is the case, we say that the class has bounded VC dimension.
 
For tournaments, these two notions even differ by at most \(1\):
\begin{lemma}\label{lem:prelims:tournament_VCdim}
If \(T\) is a tournament with \(\VCdim^+(T)=d\), then \(d-1\leq \VCdim^-(T)\leq d+1\).
\end{lemma}
\begin{proof}
Assume \(\VCdim^+(T)=d\), and let \(X\) be a maximum shattered set with respect to \(\mathcal{N}^-(T)\).

This means that the set \(\{N^-(v)\cap X\colon v\in V(T)\}\subseteq\mathcal{P}(X)\) contains all \(2^{|X|}\)
possible subsets. In particular, if we remove the in-neighborhoods of the vertices in \(X\) itself,
then the resulting set \(\mathcal{F}\coloneqq \{N^-(v)\cap X\colon v\in V(T)\setminus X\}\subseteq\mathcal{P}(X)\) still
contains at least \(2^{|X|}-|X|\) sets. By the Sauer-Shelah Lemma (\Cref{lem:prelims:sauer-shelah}), we thus find
\[2^{|X|}-|X|\leq |\mathcal{F}|\leq\sum_{i=0}^{\VCdim(\mathcal{F})}\binom{|X|}{i}.\]
But if \(\VCdim(\mathcal{F})\leq |X|-2\), then the right-hand side is at most \(2^{|X|}-|X|-1\), which would give a contradiction.
Thus, \(\VCdim(\mathcal{F})\geq |X|-1\) and in particular, there exists some subset \(Y\subseteq X\) with \(|Y|\geq |X|-1\) 
which is shattered by \(\mathcal{F}\). But note that since \(\mathcal{F}\) contains only in-neighborhoods of vertices outside \(X\),
and the out-neighborhoods of those vertices in \(X\) are exactly the complements of their in-neighborhoods in \(X\),
we find that also \(\mathcal{N}^+\) shatters \(Y\). This yields
\(|Y|\leq d\) and thus \(|X|\leq d+1\).

If we apply the same argument to the tournament \(\overline{T}\) obtained from \(T\) by reversing the direction of each arc
(which precisely exchanges the in- and out-neighborhoods), we find that also
\[\VCdim^+(T)=\VCdim^-(\overline{T})\leq \VCdim^+(\overline{T})+1=\VCdim^-(T)+1.\qedhere\]
\end{proof}
In the following, we will always only use the parameter \(\VCdim^+\), but \Cref{lem:prelims:tournament_VCdim} implies that
all our results also hold with the same asymptotic bounds for the parameter \(\VCdim^-\).

\paragraph*{Isomorphism computation and computational group theory.}
The tournament isomorphism problem asks whether two given tournaments are isomorphic.
The problem is equivalent under polynomial-time Turing reductions to the problem of computing a generating set
of the automorphism group of a given tournament (see~\cite[Lemma 6]{DBLP:conf/icalp/Schweitzer17}). This allows for the use of group-theoretic techniques in the design of algorithms for the tournament isomorphism problem.

In the context of computational group theory, a permutation group \(\Gamma\leq\Sym(\Omega)\)
is a group consisting of permutations of some finite set \(\Omega\). Such a permutation group
can be presented to an algorithm via a generating set, and we generally want algorithms to run in
polynomial time in the size of the ground set \(\Omega\) and the size of the generating set.
Similarly, when an algorithm is given a coset \(\phi\circ\Gamma \coloneqq \{\phi \cdot g\mid g\in \Gamma\} \) instead of a permutation group,
this means that the algorithm is given the bijection \(\phi\) and a generating set of the group \(\Gamma\)\footnote{As customary in the context of isomorphism problems, but diverging from standard group theory terminology, for cosets \(\phi\circ\Gamma\) we do not require that~$\phi$ is a permutation of~$\Omega$, and rather only require it to be a bijection composable with permutations of~$\Omega$.
}.

One of the most important facts that
makes the tournament isomorphism problem easier to handle than the more general graph isomorphism problem
is the fact that automorphism groups of tournaments have odd order and are thus solvable due to the Feit-Thompson theorem~\cite{FeitThompson}.
Since in this paper we will only need various consequences of solvability, we will not define
the notion here and only mention that the class of solvable groups is closed under taking subgroups, direct products, and wreath products.
Moreover, we also call a coset \(\phi\circ\Gamma\) solvable if the group \(\Gamma\) is solvable.

On solvable groups and cosets, many computational problems become polynomial-time solvable.
\begin{lemma}[{\cite{DBLP:conf/dimacs/Luks91, hypergraph_isomorphism, hypergraph_isomorphism_neuen}}]\label{lem:prelims:solvable_group_algorithms}
Given a solvable permutation group \(\Gamma\leq\Sym(\Omega)\) by generators, the following problems can be solved in polynomial time:
\begin{enumerate}
\item the problem to compute the orbit partition that \(\Gamma\) induces on the set \(\Omega\) (this is even polynomial-time solvable for arbitrary permutation groups),
\item the set-stabilizer problem: given~$M\subseteq \Omega$ compute~$\Stab_\Gamma(M)\coloneqq \{\gamma\in \Gamma\colon M^\gamma=M \}$,
\item the set-transporter problem: given~$M\subseteq \Omega$, \(M'\subseteq\Omega'\), and a bijection \(\phi\colon\Omega\to\Omega'\), compute~$\Trans_{\phi\circ\Gamma}(M\rightarrow M')\coloneqq \{\gamma\in\phi\circ\Gamma\colon M^\gamma=M' \}$,
\item the hypergraph-stabilizer problem\footnote{
	In \cite{hypergraph_isomorphism}, Miller actually proves polynomial-time solvability of the hypergraph-transporter (and thus also the hypergraph-stabilizer) problem for the class of \(\hat{\Gamma}_d\)-groups,
	which are in this case defined to be the class of groups all of whose composition factors embed into \(S_d\). However, his arguments generalize to the more general class \(\Gamma_d\) of groups all of whose nonabelian composition factors embed into \(S_d\). Since solvable groups are in particular \(\Gamma_1\)-groups, we also get polynomial-time solvability in this case.}:
	given a possibly edge-colored hypergraph~\(\mathcal{H}\subseteq\mathcal{P}(\Omega)\), compute \(\Stab_\Gamma(\mathcal{H})\coloneqq\{\gamma\in \Gamma\colon \mathcal{H}^\gamma=\mathcal{H}\}\), and
\item the hypergraph-transporter problem: given possibly edge-colored hypergraphs~\(\mathcal{H}\subseteq\mathcal{P}(\Omega)\), \(\mathcal{H}'\subseteq\mathcal{P}(\Omega')\), and a bijection \(\phi\colon\Omega\to\Omega'\), compute \(\Trans_{\phi\circ\Gamma}(\mathcal{H}\to\mathcal{H}')\coloneqq\{\gamma\in \phi\circ\Gamma\colon \mathcal{H}^\gamma=\mathcal{H}'\}\).
\item group intersection: given an arbitrary permutation group~$\Delta\leq \Sym(\Omega)$, compute the group intersection~$\Delta\cap \Gamma$,
\item coset intersection: given bijections~$\varphi,\psi \colon \Omega \rightarrow \Omega'$ and a permutation group~$\Delta\leq \Sym(\Omega)$, compute the coset intersection~$\varphi\circ\Gamma \cap \psi\circ\Delta$.
\end{enumerate}
\end{lemma}

Finally, let us define wreath products of permutation groups.
Given two permutation groups \(\Gamma\leq\Sym(\Omega)\) (called the bottom group)
and \(\Delta\leq\Sym(I)\) (called the top group),
the wreath product of \(\Gamma\) and \(\Delta\), denoted~$\Gamma\wr\Delta$,
is a permutation group \(\Gamma\wr\Delta\leq\Sym(\Omega\times I)\)
on the set \(\Omega\times I\) which we view as a disjoint union of \(|I|\)-copies of \(\Omega\)
and is generated by the following permutations
\begin{enumerate}
\item for each \(i\in I\) and \(\gamma\in\Gamma\), the permutation \(\gamma_i\) which maps \((\omega,j)\in\Omega\times I\) to \((\gamma(\omega),j)\) if \(j=i\) and to \((\omega,j)\) otherwise.
\item for each \(\delta\in\Delta\), the permutation which maps \((\omega,i)\in\Omega\times I\) to \((\omega,\delta(i))\).
\end{enumerate}
Thus, the group \(\Gamma\wr\Delta\) acts on \(\Omega\times I\) by independently permuting each copy of \(\Omega\) via \(\Gamma\), and in addition permuting the copies of \(\Omega\) themselves via \(\Delta\).
Note that we can efficiently compute a generating set for \(\Gamma\wr\Delta\) when given generating sets for \(\Gamma\) and \(\Delta\). Note further that we can let \(\Gamma\wr\Delta\) act
on a set \(\dot\bigcup_{i\in I}\Omega_i\) that is a disjoint union of \(I\) many sets
as long as we are given a bijection \(\phi_i\colon \Omega\to \Omega_i\) for each~\(i\).

We will use the fact that wreath products of solvable groups are again solvable. The general situation in which
one uses wreath products in the tournament isomorphism setting is the following.
Assume we have a tournament \(T\) and we have computed some isomorphism-invariant partition
\(\mathcal{P}\) of \(V(T)\). Then we can recursively compute the set \(\Aut(T[P])\) for each part \(P\in\mathcal{P}\), and often, we can also compute some sort of quotient tournament \(T/\mathcal{P}\)
on vertex set \(\mathcal{P}\) such that every automorphism of \(T\) also induces an automorphism
of \(T/\mathcal{P}\). We can also recursively compute the automorphism group \(\Aut(T/\mathcal{P})\).
Similarly to the wreath product, the whole automorphism group \(\Aut(T)\) permutes the parts of \(\mathcal{P}\) and the parts themselves can be permuted using the groups \(\Aut(T[P])\).
The induced permutation action of the parts of \(\mathcal{P}\) must be contained in \(\Aut(T/\mathcal{P})\).
If all pairs of parts induce isomorphic subtournaments,
then \(\Aut(T)\) actually embeds into the wreath product \(\Aut(T[P])\wr\Aut(T/\mathcal{P})\). This gives us a solvable encasing group of \(\Aut(T)\), and we can finally compute \(\Aut(T)\)
as the set stabilizer of the set of arcs in this encasing solvable group. If not all parts induce isomorphic subtournaments, we can split the whole tournament up into maximal collections of parts for which this is the case.
The automorphism group then embeds into the direct product of the wreath products for 
the respective subtournaments, which are again solvable.
An application of this general technique in our setting can be found in \Cref{lem:patched:solvable_encasing_to_iso}. \section{VC dimension and near twins}
Given a digraph \(D\), for an integer~$k\in \mathbb{N}$, we say two vertices \(v,w\in V(D)\) are \emph{\(k\)-near twins}, and write \(v\sim^k_D w\), if \(|(N^+_D(v)\symdiff N^+_D(w))|<k\). (More precisely they are~\emph{$k$-near out-twins}.) If the digraph \(D\) is clear from context, we also write \(v\sim^k w\) instead of \(v\sim^k_D w\).

The first result we will show is that in digraphs of VC dimension at most \(d\), there are always many pairs
of \(O(n^{1-\epsilon})\)-near twins. To show this, we will apply a lemma by Haussler
which states that set families of bounded VC dimension without near twins are small:
\begin{lemma}[{\cite[Theorem 1]{sphere_packing}}]\label{lem:near-twins:haussler_packing}
Let \(X\) be a finite set of size~$n$, let \(1\leq k\leq n\), and let \(\mathcal{F}\subseteq\mathcal{P}(X)\) be a family of subsets
of VC dimension at most \(d\) such that for every two distinct sets \(A,B\in\mathcal{F}\)
we have \(|A\symdiff B|\geq k\).
Then \(|\mathcal{F}|\leq e(d+1)\left(\frac{2en}{k}\right)^d\).
\end{lemma}

\begin{corollary}\label{cor:near-twins:existence}
For every \(d\in\N\) there is some \(\epsilon_d\in\Omega(\frac{1}{d})\) and some \(N_d\in 2^{O(d^2)}\)
such that every digraph \(D\) of order \(n>N_d\) with \(\VCdim^+(D)\leq d\)
contains a set \(X\) with \(|X|<n^{1-\epsilon_d}\) such that every vertex of \(D\) has an \(\left\lceil\frac{n^{1-\epsilon_d}}{4}\right\rceil\)-near twin in \(X\).
\begin{proof}
Set \(\epsilon_d\coloneqq\frac{1}{d+2}\), \(k\coloneqq\left\lceil\frac{n^{1-\epsilon_d}}{4}\right\rceil\), and \(N_d\coloneqq \left(8^{d}e^{d+1}(d+1)\right)^{d+2}\in 2^{O(d^2)}\).
Now, let \(D\) be a digraph of order \(n> N_d\) with \(\VCdim^+(D)\leq d\) and \(X\subseteq V(D)\) a maximal set of vertices which does not contain
a pair of \(k\)-near twins. Since this set is maximal, it follows that every vertex outside \(X\) must have a \(k\)-near twin in \(X\).
Thus, it remains to show that \(|X|<n^{1-\epsilon_d}\).

If we apply \Cref{lem:near-twins:haussler_packing} to the set family \(\{N^+_D(x)\colon x\in X\}\subseteq\mathcal{P}(V(D))\),
we find that
\begin{align*}
|X|& \leq e(d+1)\left(\frac{2en}{k}\right)^d\\
&\leq e(d+1)\left(\frac{8en}{n^{1-\epsilon_d}}\right)^d\\
&=8^de^{d+1}(d+1)\cdot n^{\epsilon_dd}\\
&= N_d^{\frac{1}{d+2}}\cdot n^{\frac{d}{d+2}}\\
&<n^{\frac{1}{d+2}}\cdot n^{1-\frac{2}{d+2}}\\
&= n^{1-\frac{1}{d+2}}\\
&=n^{1-\epsilon_d}.\qedhere
\end{align*}
\end{proof}
\end{corollary}

The graph~\((V(D),\sim^k_D)\) obtained by considering the near-twin relation is an undirected graph.
We call the connected components of the graph \((V(D),\sim^k_D)\) the \emph{\(k\)-near-twin components of \(D\)}.
Note that \Cref{cor:near-twins:existence} implies that the number of \(\left\lceil\frac{n^{1-\epsilon_d}}{4}\right\rceil\)-near-twin components of a digraph of VC dimension at most \(d\) is at most \(n^{1-\epsilon_d}\).

Before we continue, let us now present the main idea of our recursion:
Partitioning the vertex set \(V(T)\) into fewer than \(n^{1-\epsilon_d}\) many
\(\left\lceil\frac{n^{1-\epsilon_d}}{4}\right\rceil\)-near-twin components gives us an isomorphism-invariant partition on every tournament
of VC dimension at most \(d\). Using standard techniques involving wreath products (as described in \Cref{sec:prelims}),
it now suffices to first handle each \(\left\lceil\frac{n^{1-\epsilon_d}}{4}\right\rceil\)-near-twin component independently, and second
compute the actions that the automorphism group \(\Aut(T)\) and isomorphism coset \(\Iso(T,T')\) induce on the set of
\(\left\lceil\frac{n^{1-\epsilon_d}}{4}\right\rceil\)-near-twin components.

For the recursion on components, we show that the near-twin relation essentially
has small degree (\Cref{lem:near-twins:out-degree,lem:near-twins:split-or-degree-bound}) which allows us to use bounded-degree
techniques that were already successfully used in the context of tournament isomorphism in \cite{tournament_isomorphism_tww,DBLP:journals/adam/ArvindPR25}.

In order to compute the actions that \(\Aut(T)\)  and \(\Iso(T,T')\) induce on the set of \(\left\lceil\frac{n^{1-\epsilon_d}}{4}\right\rceil\)-near-twin components,
we try to contract each near-twin component into a single vertex to obtain a quotient tournament and compute
isomorphism cosets between these quotient tournaments of size less than \(n^{1-\epsilon_d}\).
However, to make this work, we will need to address a crucial issue:
Since the \(\left\lceil\frac{n^{1-\epsilon_d}}{4}\right\rceil\)-near-twin components are usually not twin classes with respect to their complement,
it is not immediately clear how to define the edge set of this quotient tournament. A common trick used to deal with this issue
is to pick as the arc direction between two components the majority arc direction between these components,
which exists whenever the components have odd size (while components of even size can always be decomposed using degree arguments).
However, we do not know whether this tournament again has bounded VC dimension and hence,
we instead choose to delete some arcs where the arc direction is not clear.
We will achieve this via the notions of chunks and patched tournaments introduced in the next sections.

Going back to the first task at hand, we argue that the \(\left\lceil\frac{n^{1-\epsilon_d}}{4}\right\rceil\)-near-twin relation essentially has small degree.
We will need the bound for oriented digraphs, where it applies to the out-neighborhood of a vertex. A very similar bound for tournaments already appears in \cite[Lemma 3.2]{tournament_isomorphism_tww}.
\begin{lemma}\label{lem:near-twins:out-degree}
For every oriented digraph \(D\), every vertex \(v\in V(D)\) and every \(k\geq 2\), the number of \(k\)-near twins of~\(v\) in~\(N^+(v)\) is at most \(2k-3\).
\end{lemma}
\begin{proof}
Assume the claim was false. Let \(v\in V(D)\) be a vertex with at least \(2k-2\) many \(k\)-near twins
among its out-neighbors, and let \(N^+_{\sim^k}(v)\coloneqq\{w\in N^+_{D}(v)\colon v\sim^k w\}\)
be the set of those \(k\)-near-twins.

Then, if we consider the subdigraph of \(D\) induced by \(N^+_{\sim^k}(v)\),
there exists a vertex \(w\in N^+_{\sim^k}(v)\) of out-degree at most \(\left\lfloor\frac{|N^+_{\sim^k}(v)|-1}{2}\right\rfloor\).
But since all vertices in \(N^+_{\sim^k}(v)\) are out-neighbors of \(v\),
the out-neighborhoods of \(v\) and \(w\) in \(D\) differ by at least
\[\left\lceil\frac{|N^+_{\sim^k}(v)|+1}{2}\right\rceil
\geq \left\lceil\frac{2k-1}{2}\right\rceil
\geq k\] vertices. But this contradicts the assumption that \(v\) and \(w\) are \(k\)-near twins.
\end{proof}

Even though the above lemma only restricts the number of \(k\)-near twins among the out-neighbors of some vertex,
we can use this to also restrict the number of \(k\)-near twins among the in-neighbors of each vertex or alternatively find an isomorphism-invariant partition.
\begin{lemma}\label{lem:near-twins:split-or-degree-bound}
Let \(D\) be an oriented digraph and \(k\geq 2\).
Then either, every vertex has at most \(4k-6\) many \(k\)-near twins among its neighbors \(N^+(v)\cup N^-(v)\),
or we can compute in polynomial time a nontrivial, isomorphism-invariant vertex-coloring on \(D\).
\end{lemma}
\begin{proof}
Let \(E_k\coloneqq\{vw\in E(D)\colon v\sim^k_D w\}\) be the set of arcs induced by the \(k\)-near-twin relation,
and \(D_k\coloneqq (V(D),E_k)\) the digraph induced by this arc set.
By \Cref{lem:near-twins:out-degree}, this digraph has out-degree at most \(2k-3\). But since the sum of all out-degrees
in this digraph must equal the sum of all in-degrees, this implies that also the average in-degree is bounded by \(2k-3\).
If indeed all in-degrees are bounded by \(2k-3\), we are done. Otherwise, we color each vertex by its in-degree in \(D_k\),
which yields a nontrivial, isomorphism-invariant vertex-coloring.
\end{proof}
Note that the degree bound in \Cref{lem:near-twins:split-or-degree-bound} implies in particular that if \(\VCdim^+(D)\leq d\),
then every vertex has fewer than \(n^{1-\epsilon_d}\) many \(\left\lceil\frac{n^{1-\epsilon_d}}{4}\right\rceil\)-near twins among its neighbors.
If the outcome is a nontrivial vertex-coloring, this will later allow us to treat each color class separately\footnote{This is essentially a form of orbit-by-orbit processing.} via \Cref{cor:patched:nontrivial_coloring}. \section{Patched tournaments}
During our reductions in the following sections, it will sometimes happen that we want
to recursively compute the isomorphism cosets between two substructures, where these substructures
are no longer tournaments. However, the parts where these substructures fail to be tournaments
will always come with an already computed solvable group action on these parts.
In this section, we make this concept formal via the notion of \emph{patched tournaments}.

A \emph{patched tournament} is a \(4\)-tuple \(T=\left(D,\mathcal{P}, \left(\Iso_T(P,P')\right)_{P,P'\in \mathcal{P}},Q\right)\) such that
\begin{enumerate}
\item \(D\) is an oriented digraph,
\item \(\mathcal{P}\) is a partition of \(V(D)\) which we call the \emph{patch partition} and whose parts we call \emph{patches} such that
	\begin{enumerate}
	\item for each patch \(P\in\mathcal{P}\), the induced digraph \(D[P]\) is an independent set,
	\item for each pair of distinct vertices \(x,y\in V(D)\) which do not belong to the same patch, either \(xy\in E(D)\) or \(yx\in E(D)\),
	\end{enumerate}
\item for every two patches \(P,P'\in\mathcal{P}\), \(\Iso_T(P,P')\) is a (possibly empty) set of bijections \(P\to P'\) such that the family
	\(\left(\Iso_T(P,P')\right)_{P,P'\in \mathcal{P}}\) defines a solvable groupoid on the set of patches:
	\begin{enumerate}
	\item for every \(\phi\in\Iso_T(P,P')\) and \(\psi\in\Iso_T(P',P'')\), we have \(\psi\circ\phi\in\Iso_T(P,P'')\),
	\item for every \(\phi\in\Iso_T(P,P')\), we have \(\phi^{-1}\in\Iso_T(P',P)\),
	\item for every \(P\in\mathcal{P}\), the set \(\Aut_T(P)\coloneqq\Iso_T(P,P)\) is a solvable group under composition,
	\end{enumerate}
\item \(Q\) is a tournament on vertex set \(\mathcal{P}\), called the \emph{patch tournament}.
\end{enumerate}
Given a patched tournament \(T\), we also write \(V(T)\) and \(E(T)\), \(\mathcal{P}(T)\) and \(Q(T)\)
for the underlying vertex set, arc set, patch partition and patch tournament.

Note that the groupoid structure can be compactly represented by giving a generating set of the group \(\Aut_T(P)\) for each patch \(P\),
and giving a single bijection \(\phi_{P,P'}\in\Iso_T(P,P')\) for every two distinct parts \(P,P'\in\mathcal{P}\) for which \(\Iso_T(P,P')\) is non-empty.
The full, non-empty sets \(\Iso_T(P,P')\) can then be reconstructed as \(\Iso_T(P,P')=\phi_{P,P'}\circ\Aut_T(P)\).

We can also consider vertex- or arc-colored patched tournaments, where we additionally require all bijections in the sets \(\Iso_T(P,P')\)
to preserve vertex-colors (arc-colors are preserved automatically since there are no arcs within patches).

A \emph{compatible class of patched tournaments} is a class \(\mathcal{C}\) of patched tournaments
together with a family \(\left(\Iso_{T,T'}(P,P')\right)_{P\in \mathcal{P}(T),P'\in\mathcal{P}(T')}\)
for every two patched tournaments \(T,T'\in\mathcal{C}\)
which extends the solvable groupoid structures on \(\mathcal{P}(T)\) for every \(T\in\mathcal{C}\) to a solvable groupoid structure
on \(\dot\bigcup_{T\in\mathcal{C}}\mathcal{P}(T)\).

Note that we can interpret every (colored) tournament \(T\) as a patched tournament by taking the discrete vertex partition for \(\mathcal{P}\),
letting \(\Iso_T(\{v\},\{w\})\) be the set of (vertex-color preserving) bijections \(\{v\}\to\{w\}\) and setting \(Q(T)\coloneqq T\).
Similarly, we can interpret every class of (colored) tournaments as a compatible class of (colored) patched tournaments.

Next, we want to define substructures of patched tournaments.
Given a patched tournament \(T\) and a subset \(X\subseteq V(T)\), we define the \emph{patched subtournament} induced by~$X$, denoted \(T[X]\), to be the tournament induced by~$X$ with
\begin{enumerate}
\item patch partition \(\mathcal{P}[X]\coloneqq\{P\cap X\colon P\in\mathcal{P}\}\setminus\{\emptyset\},\)
\item isomorphism sets \(\Iso_{T[X]}(P\cap X,P'\cap X)\coloneqq\{\phi|_{P\cap X}\colon \phi\in\Trans_{\Iso_T(P,P')}(P\cap X\to P'\cap X)\}\),
\item patch tournament \(\left(\mathcal{P}[X],\{(P\cap X,P'\cap X)\in\mathcal{P}[X]^2\colon PP'\in E(Q(T))\}\right)\)
\end{enumerate}
Note that the operation of taking patched subtournaments is transitive, i.e., if \(X\subseteq Y\subseteq V(T)\),
then \(T[Y][X]=T[X]\).
Note further that because the coset \(\Iso_T(P,P')\) is solvable for every pair of patches \(P\) and \(P'\), the patched subtournament \(T[X]\)
can be computed in polynomial time by \Cref{lem:prelims:solvable_group_algorithms}.
Furthermore, using the analogous definition to define isomorphism sets \(\Iso_{T[X],T'[Y]}(P\cap X,P'\cap Y)\),
the subtournaments of patched tournaments from a compatible class again form a compatible class.

\subsection{Isomorphisms of patched tournaments}
Next, we define isomorphisms of patched tournaments.
If \(T\) and \(T'\) form compatible (colored) patched tournaments, then a bijection \(\phi\colon V(T)\to V(T')\) is an isomorphism if
\begin{enumerate}
\item \(\phi\) is an isomorphism of the underlying (colored) digraphs \((V(T),E(T))\) and \((V(T'),E(T'))\),
\item \(\phi\) preserves the patch partition, i.e., \(\mathcal{P}(T)^\phi=\mathcal{P}(T')\),
\item for every patch \(P\in\mathcal{P}(T)\) with image patch \(\phi(P)\in\mathcal{P}(T')\), we have \(\phi|_P\in\Iso_{T,T'}(P,\phi(P))\), and
\item \(\phi\) induces an isomorphism of patch tournaments \(Q(T)\to Q(T')\).
\end{enumerate}
We write \(\Iso(T,T')\) for the set of isomorphisms from \(T\) to \(T'\) and write \(\Aut(T)\coloneqq\Iso(T,T)\) for the set of automorphisms of \(T\).
Note that if \(T\) and \(T'\) are (colored) tournaments interpreted as patched tournaments, then an isomorphism of patched tournaments is exactly
an isomorphism of tournaments in the usual sense. Thus, the tournament isomorphism problem reduces to the patched tournament isomorphism problem.

It follows directly from the definition that this notion of isomorphisms is well-behaved in the following sense.
\begin{lemma}\label{lem:patched:local-isos_restriction_composition_inverses}
Isomorphisms between patched tournaments in a compatible class are preserved under composition, taking inverses and restriction.
In particular, each set \(\Aut(T)\) forms a group, and each of the sets \(\Iso(T,T')\) is either empty or a coset of \(\Aut(T)\).
\end{lemma}

Furthermore, our new notion is restrictive enough to still yield solvability of all groups in question.
\begin{lemma}\label{lem:patched:local-autos_solvable}
For every patched tournament \(T\), the group \(\Aut(T)\) is solvable.
Similarly, for every pair of compatible patched tournaments \(T\) and \(T'\),
the coset \(\Iso(T,T')\) is either empty or solvable.
\end{lemma}
\begin{proof}
Because every automorphism of \(T\) preserves the patch partition,
\(\mathcal{P}(T)\) forms a block system for \(\Aut(T)\).
Let \(\mathcal{O}\) be the partition of \(V(T)\) which is induced by the orbit partition
of \(\Aut(T)\) on \(\mathcal{P}(T)\).
Then \(\Aut(T)\leq\prod_{O\in\mathcal{O}}\Aut(T[O])\) and since subgroups and products of solvable groups
are again solvable, it suffices to prove that all of the groups \(\Aut(T[O])\) are solvable.

For this, note that every part \(O\in\mathcal{O}\) is a union of patches, and denote
the solvable permutation group that the group \(\Aut(T[O])\) induces on the set of patches
\(\mathcal{P}[O]\coloneqq \{P\in\mathcal{P}\colon P\subseteq O\}\) by \(\Aut(T[O])^\mathcal{P}\).
Note that by the definition of the partition \(\mathcal{O}\), the group \(\Aut(T[O])^{\mathcal{P}}\)
is transitive on the set \(\mathcal{P}[O]\), and for every two patches \(P,P'\in\mathcal{P}[O]\),
we find some isomorphism \(\phi_{P,P'}\colon T[P]\to T[P']\). Thus, for every patch \(P\subseteq O\),
the group \(\Aut(T[O])\) embeds into the wreath product \(\Aut(T[P])\wr\Aut(T[O])^\mathcal{P}\).
But \(\Aut(T[P])=\Aut_T(P)\) is solvable by the definition of patched tournaments.
The group \(\Aut(T[O])^\mathcal{P}\) embeds into \(\Aut(Q(T[O]))\), and this group is the automorphism group of a tournament
and hence solvable. Thus, \(\Aut(T[O])\) embeds into a wreath product of solvable groups and is thus also solvable.

Similarly, if \(\Iso(T,T')\) is non-empty, then \(\Iso(T,T')=\phi\circ\Aut(T)\) for some isomorphism \(\phi\colon T\to T'\).
But since the group \(\Aut(T)\) is solvable, so is the coset \(\Iso(T,T')\).
\end{proof}

Finally, we note that the isomorphism coset between compatible patched tournaments on a single patch is trivially polynomial-time computable.
\begin{lemma}\label{lem:patched:iso_single_patch}
Let \(T\) and \(T'\) be compatible patched tournaments such that \(\mathcal{P}(T)\) and \(\mathcal{P}(T')\) both contain just a single patch.
Then \(\Iso(T,T')\) is computable in polynomial time.
\end{lemma}
\begin{proof}
We have \(\Iso(T,T')=\Iso_{T,T'}(V(T),V(T'))\) which is given as part of the compatibility structure.
\end{proof}

Next, we will state the two main lemmas that we will use to compute isomorphism cosets between patched tournaments whenever we have already
found an encasing solvable coset.
\begin{lemma}\label{lem:patched:solvable_encasing_to_iso}
Given two compatible colored patched tournaments \(T\) and \(T'\)
and a solvable coset \(\phi\circ\Gamma\) of bijections \(V(T)\to V(T')\) satisfying~\( \phi\circ\Gamma\geq\Iso(T,T')\),
we can compute a generating set for \(\Iso(T,T')\) in polynomial time.
\end{lemma}
\begin{proof}
We compute the coset~$\psi\circ\Delta$ of those bijections~$\tau\colon V(T)\rightarrow V(T')$ for which for each patch~$P \in \mathcal{P}(T)$ we have that~$\tau(P)$ is a patch~$P'$ of~$T'$ and
have~$\tau|_{P}\in \Iso(T[P],T'[P'])$. This coset~$\psi\circ\Delta$ is obtained as follows. To compute the bijection~$\psi$ we pick an arbitrary bijection~$\overline{\psi}$ from the set of patches of~$T$ to the set of patches of~$T'$ that assigns to every patch~$P$ of~$T$ a patch~$P'=\overline{\psi}(P)$ of the same isomorphism type in~$T'$. If no such bijection exists, then~$\Iso(T,T')$ is empty and we are done.
Otherwise, we choose for each~$P$ an arbitrary bijection in~$\Iso(T[P],T'[\overline{\psi}(P)])$. The map~$\psi$ then maps vertices in~$P$ to vertices of~$\overline{\psi}(P)$ according to this bijection.
The group~$\Delta$ is generated as follows:
For each pair of distinct patches~$P_1,P_2\in \mathcal{P}(T)$ for which \(\Iso_T(P_1,P_2)\) is non-empty, pick an element~$\tau_{P_1,P_2}\in\Iso_T(P_1,P_2)$. Extend~$\tau_{P_1,P_2}$ to~$\overline{\tau_{P_1,P_2}}$ of~$V(T)$ by mapping elements of~$P_2$ back to~$P_1$ according to~$\tau_{P_1,P_2}^{-1}$ and fixing elements outside of~$P_1$ and~$P_2$.
For each patch~$P_1$, we extend the generators~$\rho_{P_1}$ of~$\Iso_T(P_1,P_1)$ to a permutation of~$\widehat{\rho_{P_1}}$ of~$V(T)$ by fixing all vertices outside of~$P_1$.
Then~$\Delta$ is generated by the elements~$\overline{\tau_{P_1,P_2}}$ over all~$P_1,P_2\in \mathcal{P}(T)$ where \(\Iso_T(P_1,P_2)\neq\emptyset\)
and the extensions of generators~$\widehat{\rho_{P_1}}$ over all~$P_1\in \mathcal{P}(T)$.
 
To compute~$\Iso(T,T')$ we compute first~$\tau\circ\Theta\coloneqq \phi\circ\Gamma\cap \psi\circ\Delta$ using Lemma~\ref{lem:prelims:solvable_group_algorithms}. Note that~$\Theta$ is solvable.
Inside~$\tau\circ\Theta$, we then compute the set transporter for the arc sets of both patched tournaments, for every corresponding set of vertex- or arc-color classes,
and for the (lifted) arc set of the underlying patch tournaments \(Q(T)\) and \(Q(T')\).
The resulting coset is then exactly \(\Iso(T,T')\).
\end{proof}

The lemma above in particular implies that we can handle the color classes of vertex-colored patched tournaments independently.
\begin{corollary}\label{cor:patched:nontrivial_coloring}
Let \(T\) and \(T'\) be two compatible arc-colored patched tournaments and let \(\chi\) and \(\chi'\) be vertex-colorings on \(T\) and \(T'\) respectively.
Then the computation of \(\Iso((T,\chi),(T',\chi'))\) reduces in polynomial time to computing \(\Iso(T[\chi^{-1}(c)],T'[(\chi')^{-1}(c)])\) for each
vertex color \(c\).
\end{corollary}
\begin{proof}
If any of the cosets \(\Iso(T[\chi^{-1}(c)],T'[(\chi')^{-1}(c)])\) is empty, then so is \(\Iso(T,T')\).
Thus, assume that for each vertex color, \(\Iso(T[\chi^{-1}(c)],T'[(\chi')^{-1}(c)])=\phi_c\circ\Aut(T[\chi^{-1}(c)])\) for some bijection \(\phi_c\).
Now, define \(\phi\colon V(T)\to V(T')\) by setting \(\phi(v)\coloneqq\phi_{\chi(v)}(v)\) and note that
\(\Iso(T,T')\leq\phi\circ\prod_{c\in\image(\chi)}\Aut(T[\chi^{-1}(c)])\). Thus, we have found a solvable encasing coset for \(\Iso(T,T')\),
which allows us to compute a generating set for \(\Iso(T,T')\) in polynomial time via \Cref{lem:patched:solvable_encasing_to_iso}.
\end{proof}

In addition to handling color classes separately, we can also handle the parts of an isomorphism-invariant partition separately
if we in addition have access to an appropriate solvable permutation group on the set of parts of the partition.
 
\begin{lemma}\label{lem:patched:iso_via_block_system}
Let \(T\) and \(T'\) be two compatible vertex- and arc-colored patched tournaments and \(\mathcal{B}\) and \(\mathcal{B}'\) two isomorphism-invariant partitions
of \(V(T)\) and \(V(T')\) respectively. 
Assume that for all parts \(B\in\mathcal{B}\) and \(B'\in\mathcal{B}'\), we are given the isomorphism coset \(\Iso(T[B],T'[B'])\),
and further, we are given a (possibly empty) solvable coset \(\phi_\mathcal{B}\circ\Gamma_\mathcal{B}\) of bijections \(\mathcal{B}\to\mathcal{B}'\) such that every isomorphism
\(\phi\colon T\to T'\) induces a bijection \(\phi|_\mathcal{B}\colon\mathcal{B}\to\mathcal{B}'\) within \(\phi_\mathcal{B}\circ\Gamma_\mathcal{B}\).
Then we can compute \(\Iso(T,T')\) in polynomial time.
\end{lemma}
\begin{proof}
We want to use the given data to either directly decide that \(\Iso(T,T')\) is empty, or compute a solvable coset \(\phi\circ\Gamma\geq\Iso(T,T')\).
In the latter case, the claim then follows from \Cref{lem:patched:solvable_encasing_to_iso}.

For all pairs of parts for which \(T[B]\) and \(T'[B']\) are isomorphic, we write \(\Iso(T[B],T'[B'])=\phi_{B,B'}\circ\Aut(T[B])\).
If for some part \(B\in\mathcal{B}\), all of the sets \(\Iso(T[B],T'[B'])\) with \(B'\in\mathcal{B}'\) are empty, then so is \(\Iso(T,T')\) and we are done.
Thus, we may assume otherwise, which in particular allows us to extract a generating set of \(\Aut(T[B])\) for each part \(B\in\mathcal{B}\).

If we have two parts \(B_1,B_2\in\mathcal{B}\), then every isomorphism \(T[B_1]\to T[B_2]\) factors through every part \(B'\in\mathcal{B}'\)
for which \(\Iso(T[B_1],T'[B'])\) is non-empty. Thus, we can also compute the coset \(\Iso(T[B_1],T[B_2])\) for every two parts \(B_1,B_2\in\mathcal{B}\),
and similarly the coset \(\Iso(T'[B_1'],T'[B_2'])\) for every two parts \(B_1',B_2'\in\mathcal{B}'\).

We color the parts in~$\mathcal{B}$ and~$\mathcal{B}'$ according to their isomorphism type. Then, using the colored hypergraph transporter problem, we compute the subcoset \(\phi'_\mathcal{B}\circ\Gamma'_\mathcal{B}\leq \phi_\mathcal{B}\circ\Gamma_\mathcal{B}\)
of those bijections that send parts to parts of the same isomorphism type. If this coset is instead empty, then so is \(\Iso(T,T')\) and we are again done. So assume otherwise.

Then, we compute the coset~$\psi\circ\Delta$ of those bijections that map each element of~$B\in \mathcal{B}$ to an element~$B'\in \mathcal{B}'$ so the restriction to~$B$ is contained in~$\Iso(T[B],T'[B'])$ and so that the induced map~$\mathcal{B}\to\mathcal{B}'$ is contained in \(\phi_\mathcal{B}'\circ\Gamma_\mathcal{B}'\). This coset can be computed similarly to the construction at the beginning of the proof of Lemma~\ref{lem:patched:solvable_encasing_to_iso} with the partitions~\(\mathcal{B}\) and \(\mathcal{B}'\) replacing the patch partitions. 

The coset~$\psi\circ\Delta$ is solvable.   
Thus, the claim now follows from \Cref{lem:patched:solvable_encasing_to_iso}.
\end{proof}

In particular, the lemma above allows us to carry over the \(n^{O(\log n)}\) tournament isomorphism algorithm from \cite{DBLP:conf/stoc/BabaiL83} to patched tournaments.
\begin{lemma}\label{lem:patched:patched_tournament_iso_n^logn}
Given two compatible colored patched tournaments \(T\) and \(T'\) of order \(n\), we can compute \(\Iso(T,T')\) in time \(n^{O(\log n)}\).
\end{lemma}
\begin{proof}
If we apply \Cref{lem:patched:iso_via_block_system} to the patch partitions \(\mathcal{P}(T)\) and \(\mathcal{P}(T')\) themselves,
it only remains to compute a solvable coset of bijections \(\mathcal{P}\to\mathcal{P}'\) which contains every bijection induced by an isomorphism.
But such a coset is given by the coset \(\Iso(Q(T),Q(T'))\) of isomorphisms between the patch tournaments of \(T\) and \(T'\) respectively.
Using the \(n^{O(\log(n))}\) algorithm from \cite{DBLP:conf/stoc/BabaiL83}, we can compute this coset in time \(|\mathcal{P}|^{O(\log|\mathcal{P}|)}\leq n^{O(\log n)}\).
\end{proof}

\subsection{Patched tournaments and VC dimension}

Given a patched tournament \(T\) and a partition \(\mathcal{Q}\) of \(V(T)\) which is a coarsening of the patch partition \(\mathcal{P}(T)\),
we write \(T-\mathcal{Q}\) for the digraph obtained from (the underlying digraph of) \(T\) by deleting all arcs within each part of \(\mathcal{Q}\). If we can additionally define on \(T-\mathcal{Q}\) the structure of a patched tournament again,
we say that \(T-\mathcal{Q}\) is \emph{obtained by patching \(T\) with \(\mathcal{Q}\)}.
Importantly, our recursion will repeatedly compute a coarser patch partition along with isomorphism cosets for every pair of new patches
and then patch the old patched tournament with the new set of patches. In order to be able to apply this strategy, however,
we need to make sure that the VC dimension (that is, the VC dimension of the underlying digraph) does not repeatedly increase while patching.
To do this, we say that the \emph{patched VC dimension \(\pVCdim^+(T)\)} of a patched tournament \(T\) is the maximal VC dimension
of all patched tournaments obtained by patching \(T\). One of the crucial combinatorial insights of our algorithm is that the
patched VC dimension of some patched tournament is bounded in its ordinary VC dimension.
\begin{lemma}\label{lem:patched:VC-dim_bound}
Let \(T\) be a patched tournament with \(\VCdim^+(T)\leq d\), and \(\mathcal{Q}\) a coarsening of \(\mathcal{P}(T)\).
Then \(\VCdim^+(T-\mathcal{Q})\leq O(d\log d)\).
\end{lemma}
\begin{proof}
Assuming otherwise, let \(X\) be a shattered set in \(T-\mathcal{Q}\)
such that \((d+1)\log_2(|X|+1)<|X|\), which exists of size at most \(O(d\log d)\).

Then for every vertex \(v\in V(T)\), contained in some part \(P\in\mathcal{Q}\) say, we get
\(N^+_{T-\mathcal{Q}}(v)=N^+_T(v)\setminus P\) and thus in particular \(N^+_{T-\mathcal{Q}}(v)\cap X=(N^+_T(v)\cap X)\setminus P\).
It follows that
\begin{align*}
2^{|X|}
&= \left|\{N^+_{T-\mathcal{Q}}(v)\cap X\colon v\in V(T)\}\right|\\
&\leq (|X|+1)\left|\{N^+_T(v)\cap X\colon v\in V(T)\}\right|\\
	\intertext{since there are at most \(|X|\) choices of a part \(P\) which intersects \(X\), and in all other cases \(N^+_{T-\mathcal{Q}}(v)\cap X=N^+_T(v)\cap X\).
	By the Sauer-Shelah Lemma (\Cref{lem:prelims:sauer-shelah}), the right-hand side of the above inequality is bounded by}
&\leq (|X|+1)(|X|^d+1)\\
&\leq (|X|+1)^{d+1}\\
&\leq 2^{(d+1)\log_2(|X|+1)}\\
&<2^{|X|},
\end{align*}
which is a contradiction.
\end{proof}
In particular, \Cref{lem:patched:VC-dim_bound} implies that if a tournament \(T\) has VC dimension at most \(d\),
then all patched tournaments obtained by repeatedly passing to subtournaments or patching have bounded patched VC dimension
and thus also bounded VC dimension. Since all recursive subinstances constructed by our algorithm will be of this form,
this makes sure that the VC dimension of all instances always stays bounded.

By restating \Cref{cor:near-twins:existence} in terms of patched tournaments and VC dimension and combining this with the degree bound in \Cref{lem:near-twins:split-or-degree-bound}, we get the following corollary.
\begin{corollary}\label{cor:patched:near-twin-structure}
For every \(d\in\N\), there is some \(\epsilon_d\in \Omega(\frac{1}{d})\) and some \(N_d\in 2^{O(d^2)}\)
such that for every patched tournament \(T\) of order \(n>N_d\) with \(\pVCdim^+(T)\leq d\),
at least one of the following is the case:
\begin{enumerate}
\item we can compute in polynomial time a nontrivial, isomorphism-invariant vertex-coloring of \(T\), or
\item for \(k\coloneqq \left\lceil\frac{n^{1-\epsilon_d}}{4}\right\rceil\), \(T\) decomposes into fewer than \(n^{1-\epsilon_d}\) many \(k\)-near-twin components, and further, every vertex \(v\in V(T)\) has fewer than \(n^{1-\epsilon_d}\) many \(k\)-near twins outside its own patch.
\end{enumerate}
\end{corollary}
\begin{proof}
Since for every patched tournament \(T\) we have \(\VCdim^+(T)\leq\pVCdim^+(T)\), we can take \(\epsilon_d\) and \(N_d\) as in \Cref{cor:near-twins:existence}.
This guarantees that there exists a set \(X\) of size less than \(n^{1-\epsilon_d}\) such that every vertex has some \(k\)-near twin within this set.
But since such a set must contain at least one vertex from each \(k\)-near-twin component,
this implies that there are fewer than \(n^{1-\epsilon_d}\) many \(k\)-near-twin components.

Finally, \Cref{lem:near-twins:split-or-degree-bound} applied to \(T\) and \(k\)
guarantees that either each vertex has at most \(4k-6<n^{1-\epsilon_d}\) many \(k\)-near twins outside its own patch,
or we find a nontrivial isomorphism-invariant vertex coloring of \(T\).
\end{proof}
Note that a nontrivial isomorphism-invariant vertex-coloring will allow us to treat the color classes separately by \Cref{cor:patched:nontrivial_coloring}.
Thus, \Cref{cor:patched:near-twin-structure} essentially states that the isomorphism problem for patched tournaments of bounded patched VC dimension
reduces to patched subtournaments that have either bounded size or admit a decomposition as in the second case of \Cref{cor:patched:near-twin-structure}. \section{Computing isomorphism of patched tournaments of bounded VC dimension}
From now on we always assume that we have a pair of compatible arc-colored patched tournaments
\(T\) and \(T'\) of the same order \(n\) and \(\pVCdim^+(T)\leq d\).
Note that the only properties of bounded patched VC dimension we will need are the consequences
of \Cref{cor:patched:near-twin-structure}. Thus, we do not need to assume that also \(\pVCdim^+(T')\leq d\),
since our algorithm will run in polynomial time if \Cref{cor:patched:near-twin-structure},
 applies to \(T'\), and we can immediately return \(\Iso(T,T')=\emptyset\) if it does not.

For this whole section, we assume that we are in the second case of \Cref{cor:patched:near-twin-structure} (since otherwise we can immediately recurse on the color classes via \Cref{cor:patched:nontrivial_coloring}) and fix the constants \(\epsilon_d>0\) and \(N_d\in\N\) from \Cref{cor:patched:near-twin-structure}
as well as \(k\coloneqq\left\lceil\frac{n^{1-\epsilon_d}}{4}\right\rceil\).
Specifically, we know that \(T\) decomposes into less than \(n^{1-\epsilon_d}\) many \(k\)-near-twin components and furthermore,
each vertex has less than \(n^{1-\epsilon_d}\) many \(k\)-near twins outside its own patch.

We now start by defining the notion of \emph{chunks}. These will later become the new patches for recursive calls.
Given a set of vertices \(X\subseteq V(T)\), we say that a vertex \(u\in V(T)\setminus X\) is \emph{close to \(X\)} if either
\begin{enumerate}
\item \(u\) is in the same patch as some vertex in \(X\),
\item \(u\) is a \(k\)-near twin of some vertex in \(X\),
\item there exist two vertices \(v,w\in X\) such that \(|(N^+(v)\symdiff N^+(w))\setminus X|<k\)
	and \(u\in N^+(v)\symdiff N^+(w)\).
\end{enumerate}

We set \(X^{(0)}\coloneqq X\) and then inductively define
\[X^{(i+1)}\coloneqq X^{(i)}\cup\{u\in V(T)\setminus X^{(i)}\colon u \text{ is close to } X^{(i)}\}.\]
Since this process defines an increasing chain \(X=X^{(0)}\subseteq X^{(1)}\subseteq \dots\subseteq V(T)\),
it must stabilize after at most \(n\) steps and we set \(X^{(\infty)}\coloneqq\bigcup_{i\in\N} X^{(i)}\) to be the
set with which this process stabilizes.

We first collect a few easy properties of the closure operation we just defined.
\begin{lemma}\label{lem:properties_of_closure}
Let \(T\) be a patched tournament, \(X\subseteq V(T)\) an arbitrary set, and \(k\) as above. Then
\begin{enumerate}
\item\label{lempart:monotone} for every \(Y\subseteq X\) we have \(Y^{(i)}\subseteq X^{(i)}\) for all \(i\in\N\cup\{\infty\}\),
\item\label{lempart:union_of_near_twin_components} \(X^{(\infty)}\) is a union of \(k\)-near-twin components of \(T\),
\item\label{lempart:union_of_patches} \(X^{(\infty)}\) is a union of patches,
\item\label{lempart:twin_classes} if \(v,w\in X^{(\infty)}\) belong to the same \(k\)-near-twin component of \(T\), they have the same in- and out-neighborhoods in \(V(T)\setminus X^{(\infty)}\).
\end{enumerate}
\begin{proof}
The first part is immediate by induction, since every vertex \(v\in V(T)\setminus X\) which is close to \(Y\) is also close to \(X\).
The second and third claim are also immediate since~\(X^{(i+1)}\) contains all \(k\)-near twins of vertices in \(X^{(i)}\)
and vertices within the same patch as some vertex in~\(X^{(i)}\) are also contained in~\(X^{(i+1)}\).

The fourth claim follows since for every \(i\in\N\), whenever the out-neighborhoods of \(v,w\in X^{(i)}\) outside \(X^{(i)}\) differ in less than \(k\) vertices,
then \(X^{(i+1)}\) contains~\((N^+(v)\symdiff N^+(w))\). Thus, if \(C\) is a \(k\)-near-twin component within \(X^{(\infty)}\), then any two vertices \(u,v\in C\) have the same out-neighborhood in \(V(T)\setminus X^{(\infty)}\). But since \(X^{(\infty)}\) is a union of patches, there are no non-edges crossing the cut \((X^{(\infty)},V(T)\setminus X^{(\infty)})\).
Thus, if \(u,v\in C\) have the same out-neighborhood in \(V(T)\setminus X^{(\infty)}\), then they also have the same in-neighborhood in \(V(T)\setminus X^{(\infty)}\).
\end{proof}
\end{lemma}

\begin{definition}
A \emph{chunk} of \(T\) is the set of the form \(\{v\}^{(\infty)}\), that is, a set resulting from the process described above when starting with a set~$X= \{v\}$ containing a single vertex.	
\end{definition}
Note that the image of a chunk under any automorphism of \(T\) is again a chunk. We now describe how to construct a partition of \(V(T)\) using the chunks.
Indeed, we will show that the set of chunks either forms an isomorphism-invariant partition of \(V(T)\), which we can use to recurse,
or we obtain a nontrivial, isomorphism-invariant vertex-coloring on \(V(T)\) which we can use to recurse by \Cref{cor:patched:nontrivial_coloring}.
\begin{lemma}\label{lem:chunk_decomposition}
Let \(T\) be an arc-colored patched tournament with \(\pVCdim^+(T)\leq d\). Then at least one of the following is true:
\begin{enumerate}
\item We can compute in polynomial-time a nontrivial, isomorphism-invariant vertex-coloring on \(T\), or
\item the set of chunks forms an isomorphism-invariant partition of \(V(T)\) with the following properties:
	\begin{enumerate}
	\item every two chunks have the same size,
	\item each chunk consists of an odd number of patches of \(T\), and
	\item for each chunk \(C\), each vertex \(v\in C\) has fewer than \(n^{1-\epsilon_d}\) many \(\sim^k_{T[C]}\)-neighbors (as opposed to only \(\sim^k_{T}\)-neighbors) outside its own patch in \(T[C]\).
	\end{enumerate}
\end{enumerate}
\end{lemma}
\begin{proof}
We define the vertex-coloring \(\chi\) by setting
\(\chi(v)\coloneqq\left|\{v\}^{(\infty)}\right|\) for each vertex \(v\in\{v\}^{(\infty)}\).
If this coloring is not constant, we have found a nontrivial coloring.
Thus, assume \(\chi\) is constant. This implies that all chunks have the same size.
Now, consider two chunks \(\{v\}^{(\infty)}\) and \(\{w\}^{(\infty)}\) that have nontrivial intersection, say \(u\in\{v\}^{(\infty)}\cap\{w\}^{(\infty)}\).
By \Cref{lem:properties_of_closure}\ref{lempart:monotone}, this implies that also \(\{u\}^{(\infty)}\subseteq \{v\}^{(\infty)}\cap\{w\}^{(\infty)}\).
But since \(\chi(u)=\chi(v)=\chi(w)\), the chunk \(\{u\}^{(\infty)}\) cannot be properly
contained in either \(\{v\}^{(\infty)}\) or \(\{w\}^{(\infty)}\), which implies
\(\{v\}^{(\infty)}=\{u\}^{(\infty)}=\{w\}^{(\infty)}\).
Thus, any two chunks with non-empty intersection coincide and hence, the set of chunks forms an isomorphism-invariant partition into parts of equal size.

Next, for each chunk \(\{v\}^{(\infty)}\) consider the tournament induced by \(Q(T)\) on the set of patches
contained in \(\{v\}^{(\infty)}\). If this tournament is not regular, we can color each patch by its out-degree in this tournament, which then yields a nontrivial isomorphism-invariant coloring of \(V(T)\).
Thus, we may assume that these tournaments are regular, which implies that they are defined on an odd number of parts.

The last condition follows from \Cref{lem:near-twins:split-or-degree-bound} applied to the patched tournament \(T[C]\):
If every vertex satisfies the desired degree bound, we are done, and otherwise, \Cref{lem:near-twins:split-or-degree-bound}
yields a nontrivial, isomorphism-invariant vertex-coloring on all chunks where the degree bound does not hold,
which also gives us a non-trivial isomorphism-invariant vertex-coloring of \(T\).
\end{proof}

Since in the case of a nontrivial vertex-coloring on \(T\) we can immediately recurse via \Cref{cor:patched:nontrivial_coloring},
we assume from now on that we are in the second case, and the set of chunks forms a partition of \(V(T)\) into parts of equal size,
each consisting of an odd number of patches.

Let \(\mathcal{C}\) and \(\mathcal{C}'\) be the partitions of \(V(T)\) and \(V(T')\) into chunks.
Before we continue to use these chunk partitions to solve the isomorphism problem, let us give
another high-level description of how chunks fit into our recursion strategy.

Since chunks were defined in an isomorphism-invariant way, we know that every isomorphism \(\phi\colon T\to T'\)
must send the partition \(\mathcal{C}\) to \(\mathcal{C}'\). 
In order to be able to use \Cref{lem:patched:iso_via_block_system} to compute the isomorphism coset \(\Iso(T,T')\) it thus remains to do two things.
First, we must compute isomorphism cosets \(\Iso(T[C],T'[C'])\) for every pair of chunks \(C\in\mathcal{C}\) and \(C'\in\mathcal{C}'\).
To do this, we observe that the definition of chunks together with the degree restriction on the \(k\)-near-twin relation in \Cref{cor:patched:near-twin-structure,lem:chunk_decomposition}
allow for the use of techniques similar to those developed by Luks to handle isomorphism of bounded-degree graphs~\cite{bounded_degree_isomorphism}
which were already applied to tournament isomorphism problem in \cite{DBLP:journals/adam/ArvindPR25}.
Details can be found in \Cref{sec:chunk_isomorphism}.

Second, we must find some solvable coset of bijections \(\mathcal{C}\to\mathcal{C}'\) which contains all bijections induced by isomorphisms.
For this, we note the following: since by \Cref{lem:properties_of_closure}\ref{lempart:twin_classes}, each \(k\)-near-twin component of \(T\) or \(T'\) forms a twin class with respect to the outside
of its chunk, \(k\)-near-twin components form twin-classes in the modified digraphs \(T-\mathcal{C}\) and \(T'-\mathcal{C}'\).
Thus, in these digraphs, we can contract each \(k\)-near-twin component to a single vertex without increasing the patched VC dimension.
Using the precomputed isomorphism cosets \(\Iso(T[C],T'[C'])\) for each pair of chunks, we can turn these digraphs into patched tournaments
on which we can recurse. Finally, we can then combine the recursive instances via \Cref{lem:patched:iso_via_block_system}. This second step
is carried out in \Cref{sec:permuting_chunks}.

We start with computing the isomorphism cosets between chunks.
\subsection{Isomorphisms between chunks}\label{sec:chunk_isomorphism}
Let \(\{v\}^{(\infty)}\subseteq V(T)\) and \(\{w\}^{(\infty)}\subseteq V(T')\) be two chunks.

In this section, we show that the computation of \(\Iso(T[\{v\}^{(\infty)}],T'[\{w\}^{(\infty)}])\) reduces to the computation of polynomially many patched tournament isomorphism instances
of size less than \(n^{1-\epsilon_d}\). To do this, we use the fact that the iterative definition of the chunks
gives us a natural DAG-structure of maximum out-degree \(n^{1-\epsilon_d}\) on each chunk.

As is common for isomorphism problems in a bounded-degree setting, we explain first how to compute the isomorphism coset when we fix one vertex in both chunks.
For this, let \(T[\{v\}^{(\infty)}]_v\) and \(T'[\{w\}^{(\infty)}]_w\) be the vertex- and arc-colored patched tournaments obtained from \(T[\{v\}^{(\infty)}]\) and \(T'[\{w\}^{(\infty)}]\)
by coloring the vertices \(v\) and \(w\) using a fresh color (i.e., individualizing them). The coset \(\Iso(T[\{v\}^{(\infty)}]_v,T'[\{w\}^{(\infty)}]_w)\) thus consists only of those isomorphisms
\(\phi\in\Iso(T[\{v\}^{(\infty)}], T'[\{w\}^{(\infty)}])\) for which \(\phi(v)=w\).
\begin{lemma}\label{lem:chunk_isomorphism_individualized}
Let \(T\) and \(T'\) be compatible arc-colored patched tournaments of the same order \(n\) with \(\pVCdim^+(T)\leq d\).

Then for any two vertices \(v\in V(T)\) and \(w\in V(T')\), the computation of a generating set of \(\Iso(T[\{v\}^{(\infty)}]_v,T'[\{w\}^{(\infty)}]_w)\) reduces
to less than \(\left|\{v\}^{(\infty)}\right|^2\) patched tournament isomorphism instances of patched VC dimension at most \(d\) and order less than \(n^{1-\epsilon_d}\).
\end{lemma}
\newcommand{\lvl}{\operatorname{lvl}}
\newcommand{\ext}{\mathrm{ext}}
\begin{proof}
We first note that if the two chunks have different sizes, we may immediately return that \(\Iso(T[\{v\}^{(\infty)}]_v,T'[\{w\}^{(\infty)}]_w)=\emptyset\).
Next, note that if two vertices in a chunk \(C\) are \(k\)-near twins in \(T\) (or \(T'\) respectively), then they are
also \(k\)-near twins in \(T[C]\) or \(T'[C]\) respectively. Similarly, if the out-neighborhoods of two vertices in \(C\)
disagree in less than \(k\) vertices in \(T\), then the same is true in \(T[C]\) or \(T'[C]\) respectively,
and moreover, each of these less than \(k\) vertices also lies in \(C\).
Thus, if we apply the closure operation used to define chunks in \(T[C]\) instead of~\(T\) (but with the same value of \(k\)), then we find
\(\{v\}^{(i)}_{T[C]}\supseteq\{v\}^{(i)}_T\) for each \(i\in\N\), where \(\{v\}^{(i)}_{T[C]}\) denotes the result of
\(i\) closure steps within \(T[C]\) while \(\{v\}^{(i)}_T\) denotes the result of \(i\) closure steps in \(T\).
Since furthermore \(\{v\}^{(i)}_{T[C]}\subseteq C=\{v\}^{(\infty)}_T\), we find that \(\{v\}^{(\infty)}_{T[C]}=\{v\}^{(\infty)}_T\).
Thus, we can perform the closure operation within \(T[C]\) instead of in~\(T\).
This has the advantage that the \(k\)-near-twin relation \(\sim^k_{T[C]}\) is an isomorphism-invariant relation in \(T[C]\),
while the relation \({\sim^k_T}|_{C^2}\) defined in \(T[C]\) might not be isomorphism-invariant
(and we would need to enforce isomorphism invariance by coloring this relation).
Thus, from now on until the end of this proof, \(\{v\}^{(i)}\) always refers to the set \(\{v\}^{(i)}_{T[C]}\) rather than the set \(\{v\}^{(i)}_T\).
Note that each vertex still has fewer than \(n^{1-\epsilon_d}\) many \(\sim^k_{T[C]}\)-neighbors outside its own patch by \Cref{lem:chunk_decomposition}.

With this out of the way, we define level functions \(\lvl_v\colon\{v\}^{(\infty)}\to\N\) and \(\lvl_w\colon\{w\}^{(\infty)}\to\N\)
by setting \(\lvl_v(u)\coloneqq\min\{i\in\N\colon u\in\{v\}^{(i)}\}\) for all \(u\in\{v\}^{(\infty)}\)
and analogously \(\lvl_w(u)\coloneqq\min\{i\in\N\colon u\in\{w\}^{(i)}\}\) for all \(u\in\{w\}^{(\infty)}\).
Since we may assume that the vertex sets of \(T\) and \(T'\)
are disjoint, we also write \(\lvl\) instead of \(\lvl_v\) or \(\lvl_w\) where the index does not matter or is clear from context.
Note that every isomorphism \(\phi\colon T[\{v\}^{(\infty)}]\to T'[\{w\}^{(\infty)}]\) with \(\phi(v)=w\) satisfies
\(\lvl(u)=\lvl(\phi(u))\) for all \(u\in\{v\}^{(\infty)}\).
Thus, we may introduce a vertex coloring on both \(T[\{v\}^{(\infty)}]_v\) and \(T'[\{w\}^{(\infty)}]_w\)
by setting \(\chi(u)\coloneqq\left(\lvl(u),\min\{\lvl(x)\colon x\text{ shares a patch with } u\}\right)\).
Note that while the introduction of this vertex-coloring does not affect the isomorphism coset
\(\Iso(T[\{v\}^{(\infty)}]_v,T'[\{w\}^{(\infty)}]_w)\), it can affect isomorphism cosets between induced subtournaments, which now need to preserve this coloring.

Now, we define the following vertex labeling \(\ell\) on \(\{v\}^{(\infty)}\) and the analogous
labeling on \(\{w\}^{(\infty)}\). Intuitively the label of a vertex in the chunk encodes the reason why it was added to the chunk and the other vertices that caused it to be added.
\begin{align*}
\ell^{\sim^k}(u)&\coloneqq\{x\in \{v\}^{\lvl(u)-1}\colon x\sim^k_{T[\{v\}^{(\infty)}]} u\},\\ 
\ell^{\symdiff}(u)&\coloneqq \{(x,y)\in \left(\{v\}^{\lvl(u)-1}\right)^2\colon |(N^+(x)\symdiff N^+(y))\setminus \{v\}^{\lvl(u)-1}|<k, u\in N^+(x)\symdiff N^+(y)\}\},\\
\ell^{\mathcal{P}}(u)&\coloneqq\{x\in\{v\}^{\lvl(u)-1}\colon x \text{ shares a patch with } u\}, \text{ and}\\
\ell(u)&\coloneqq(\lvl(u),\ell^{\sim^k}(u),\ell^{\symdiff}(u),\ell^{\mathcal{P}}(u)),
\end{align*}
where we set \(\{v\}^{(-1)}=\{w\}^{(-1)}=\emptyset\). A \emph{label class} is a maximal set of vertices with the same value under~$\ell$. Note that each label class is contained in a level.

Also note that for every isomorphism \(\phi\in\Iso(T[\{v\}^{(\infty)}]_v,T'[\{w\}^{(\infty)}]_w)\)
and for every vertex \(u\in \{v\}^{(\infty)}\), the label \(\ell(\phi(u))\)
is completely determined by the label \(\ell(u)\) and the isomorphism \(\phi|_{\lvl<\lvl(u)}\)
restricted to the previous levels. More specifically,
\begin{equation}\label{eq:aut_acts_on_labels}
\ell(\phi(u)) = (\lvl(u),\phi(\ell^{\sim^k}(u)),\phi(\ell^{\symdiff}(u)),\phi(\ell^{\mathcal{P}}(u)))\eqqcolon\phi(\ell(u)).
\end{equation}
This implies that the partition of the set of vertices on level \(i\) into label classes with respect to the labeling \(\ell\)
forms a block system for the pointwise stabilizer of the previous levels in the automorphism group \(\Aut(T[\{v\}^{(\infty)}]_v)\). 

\begin{claim}
The computation of the sets \(\Iso(T[L],T'[L'])\) for all label classes \(L\) of \(T[\{v\}^{(\infty)}]_v\) and \(L'\) of \(T'[\{w\}^{(\infty)}]_w\)
reduces to less than \(\left|\{v\}^{(\infty)}\right|^2\) patched tournament isomorphism instances of patched VC dimension at most \(d\) and order less than \(n^{1-\epsilon_d}\).
\end{claim}
\begin{claimproof}
We first note that the set \(\Iso(T[\{v\}],T'[\{w\}])\) contains only the unique bijection \(\{v\}\to\{w\}\),
and that \(\Iso(T[L],T'[L'])\) is empty when \(L\) and \(L'\) are on different levels since the label classes
have different vertex-colors under the coloring we introduced in this case. Thus, we can from now restrict ourselves
to pairs of label classes~$L,L'$ on the same level larger than \(0\).

We call a vertex \(u\) in either \(T[\{v\}^{(\infty)}]\) or \(T'[\{w\}^{(\infty)}]\) \emph{initial}
if \(\ell^{\mathcal{P}}(u)=\emptyset\), i.e., if \(u\) is a vertex on a minimal level within its patch.
Note that the vertex-coloring we introduced earlier distinguishes initial vertices from non-initial vertices.
Note further that every label class either contains no initial vertex or only initial vertices. We call the latter
label classes \emph{initial label classes}. Now, observe that every initial label class has size at most \(n^{1-\epsilon_d}\).
Indeed, every such class is either contained in the set of \(k\)-near twins of some vertex on the previous level,
or it is contained in the (small) symmetric difference of out-neighborhoods of two vertices on the previous level,
since otherwise, the vertices in this label class would not have been added by the closure process generating the chunk.
However, since the vertices in an initial class do not share a patch with an earlier level and each vertex has less than \(n^{1-\epsilon_d}\)
many \(k\)-near twins outside its own patch by \Cref{lem:chunk_decomposition}, initial label classes of the first type are small,
and initial label classes in the second case are also small by definition. Finally note that there are less than \(\left|\{v\}^{(\infty)}\right|\) label classes
on a level larger than \(0\) in each of \(T[\{v\}^{(\infty)}]_v\) and \(T'[\{w\}^{(\infty)}]\). Thus computing the isomorphism coset \(\Iso(T[L],T'[L'])\) between any two initial label classes \(L\) and \(L'\) entails the computation of sufficiently small instances as required by the claim.

Further, since the vertex-coloring we introduced differentiates initial classes from non-initial classes, we also know that \(\Iso(T[L],T'[L'])=\emptyset\)
when only one of the classes is initial. Finally, if both \(L\) and \(L'\) are non-initial classes, then both \(L\) and \(L'\) are contained within a single patch.
Thus, we can compute \(\Iso(T[L],T'[L'])\) via \Cref{lem:patched:iso_single_patch}. In total, we can thus assume that we have already computed the coset
\(\Iso(T[L],T'[L'])\) for any two label classes \(L\) and \(L'\).
\end{claimproof}

With these preparations out of the way, we can now describe how to compute the full isomorphism coset \(\Iso(T[\{v\}^{(\infty)}]_v,T'[\{w\}^{(\infty)}]_w)\)
by iteratively computing \(\Iso(T[\{v\}^{(i)}]_v,T'[\{w\}^{(i)}]_w)\) for increasing values of \(i\).
For \(i=0\), both \(\{v\}^{(0)}=\{v\}\) and \(\{w\}^{(0)}=\{w\}\) are themselves label classes and hence we have already computed
this isomorphism coset.

Now, assume that for some \(i\in\N\), we have computed a generating set for
\(\Iso(T[\{v\}^{(i)}]_v,T'[\{w\}^{(i)}]_w)\). If this coset is empty,
then so is \(\Iso(T[\{v\}^{(\infty)}]_v,T'[\{w\}^{(\infty)}]_w)\) and we are done.
Further, we know by \Cref{lem:patched:local-autos_solvable} that this coset as well as all isomorphism cosets between label classes
are solvable.

Consider the partitions \(\mathcal{B}\) and \(\mathcal{B}'\) of \(\{v\}^{(i+1)}\) and \(\{w\}^{(i+1)}\)
into the sets \(\{v\}^{(i)}\) or \(\{w\}^{(i)}\) and the label classes on level \(i+1\) respectively.
These partitions are isomorphism-invariant and for all pairs of partition classes \(B\in\mathcal{B}\) and \(B'\in\mathcal{B}'\)
we know the isomorphism coset \(\Iso(T[B],T'[B'])\). Thus, by \Cref{lem:patched:iso_via_block_system},
it suffices to compute a solvable coset of bijections \(\mathcal{B}\to\mathcal{B}'\) which contains every bijection
induced by an isomorphism. We can clearly assume that any such bijection should send \(\{v\}^{(i)}\) to \(\{w\}^{(i)}\)
thus it remains to find a coset of bijections of the label classes on level \(i+1\). For this, note that the set
\(\Iso(T[\{v\}^{(i)}]_v,T'[\{w\}^{(i)}]_w)\) acts on the set of possible labels via \Cref{eq:aut_acts_on_labels}.
Now, if for every isomorphism \(\phi\in\Iso(T[\{v\}^{(i)}]_v,T'[\{w\}^{(i)}]_w)\) and every vertex \(u\in\{v\}^{(i+1)}\setminus\{v\}^{(i)}\),
there was a vertex \(u'\in\{w\}^{(i+1)}\setminus\{w\}^{(i)}\) with label \(\ell(u')=\phi(\ell(u))\) (and the other way around), then
\(\Iso(T[\{v\}^{(i)}]_v,T'[\{w\}^{(i)}]_w)\) would induce an appropriate coset of bijections \(\mathcal{B}\to\mathcal{B}'\).
However, this might not always be the case and we must first pass to the subset of \(\Iso(T[\{v\}^{(i)}]_v,T'[\{w\}^{(i)}]_w)\)
consisting of those isomorphisms for which this is the case. Indeed, we say that an isomorphism \(\phi\in\Iso(T[\{v\}^{(i)}]_v,T'[\{w\}^{(i)}]_w)\)
\emph{locally extends} if
\[\left\{
\phi(\ell(u))\colon u\in\{v\}^{(i+1)}\setminus\{v\}^{(i)}\right\}
=
\left\{
\ell(u)\colon u\in\{w\}^{(i+1)}\setminus\{w\}^{(i)}\right\},
\]
that is, if there can at least plausibly be a label-preserving bijection.
The set of locally-extending isomorphisms forms a subcoset
\(\Iso^\ext(T[\{v\}^{(i)}]_v,T'[\{w\}^{(i)}]_w)\leq \Iso(T[\{v\}^{(i)}]_v,T'[\{w\}^{(i)}]_w)\) and we show how we can express it as a hypergraph transporter.

\begin{claim}
We can compute a generating set of \(\Iso^\ext(T[\{v\}^{(i)}]_v,T'[\{w\}^{(i)}]_w)\) in polynomial time.
\end{claim}
\begin{claimproof}
Set \(\Omega\coloneqq \{v\}^{(i)}\) and \(\Omega'\coloneqq \{w\}^{(i)}\).
Then we can consider \(\Iso(T[\{v\}^{(i)}]_v,T'[\{w\}^{(i)}]_w)\) as a coset of bijections
\(\Omega\times\{0,1\}\dotcup\Omega^2\to\Omega'\times\{0,1\}\dotcup(\Omega')^2\)
by defining \(\phi(\omega,i)\coloneqq(\phi(\omega),i)\) and \(\phi(\omega_1,\omega_2)\coloneqq(\phi(\omega_1),\phi(\omega_2))\).
Moreover, we can naturally interpret each label \(\ell(u)=(\lvl(u),\ell^{\sim^k}(u),\ell^{\symdiff}(u),\ell^{\mathcal{P}}(u))\)
as a subset of \(\Omega\times\{0,1\}\dotcup\Omega^2\) or \(\Omega'\times\{0,1\}\dotcup(\Omega')^2\)
by interpreting \(\ell(u)\) as the set \(L(u)\coloneqq \ell^{\sim^k}(u)\times\{0\}\dotcup\ell^{\mathcal{P}}(u)\times\{1\}\dotcup\ell^{\symdiff}(u)\).
Now, the coset \(\Iso^\ext(T[\{v\}^{(i)}]_v,T'[\{w\}^{(i)}]_w)\) is precisely the hypergraph transporter
\[\Trans_{\Iso(T[\{v\}^{(i)}]_v,T'[\{w\}^{(i)}]_w)}
\left(\{L(u)\colon u\in\{v\}^{(i+1)}\setminus\{v\}^{(i)}\}\to\{L(u)\colon u\in\{w\}^{(i+1)}\setminus\{w\}^{(i)}\}\right)\]
for which we can compute a generating set in polynomial time using \Cref{lem:prelims:solvable_group_algorithms}.
\end{claimproof}

The coset \(\Iso^\ext(T[\{v\}^{(i)}]_v,T'[\{w\}^{(i)}]_w)\) now finally also induces a coset of bijections \(\mathcal{B}\to\mathcal{B}'\)
by sending \(\{v\}^{(i)}\) to \(\{w\}^{(i)}\) and otherwise sending each label class \(L\) in \(\{v\}^{(i+1)}\setminus\{v\}^{(i)}\) to the unique label class with the corresponding label in \(\{w\}^{(i+1)}\setminus\{w\}^{(i)}\). Since \(\Iso^\ext(T[\{v\}^{(i)}]_v,T'[\{w\}^{(i)}]_w)\) is a subcoset of \(\Iso(T[\{v\}^{(i)}]_v,T'[\{w\}^{(i)}]_w)\), it further is solvable, and thus,
\Cref{lem:patched:iso_via_block_system} implies that we can compute a generating set of \(\Iso(T[\{v\}^{(i+1)}]_v,T'[\{w\}^{(i+1)}]_w)\) in polynomial time.

We can iterate this until \(\{v\}^{(i)}=\{v\}^{(\infty)}\). If \(\{w\}^{(i)}\neq\{w\}^{(\infty)}\),
then there is no isomorphism, and otherwise we have computed the coset
\(\Iso(T[\{v\}^{(\infty)}]_v,T'[\{w\}^{(\infty)}]_w)\).
\end{proof}

We get the following corollary:
\begin{corollary}\label{cor:chunk_isomorphism}
Let \(T\) and \(T'\) be compatible arc-colored patched tournaments with \(\pVCdim^+(T)\leq d\) and both of order \(n\geq N_d\). 
Assume further that both patched tournaments are partitioned into chunks of the same size.
Then the computation of the sets \(\Iso(T[C],T'[C'])\) for all chunks \(C\) of \(T\) and \(C'\) of \(T'\) reduces in polynomial time to less than \(n^3\) instances of patched tournament isomorphism, each of order less than \(n^{1-\epsilon_d}\).
\end{corollary}

\begin{proof}
Note that if \(C\) and \(C'\) are chunks and \(v\in C\) is an arbitrary vertex, then 
\(\Iso(T[C],T'[C'])=\bigcup_{w\in C'}\Iso(T[C]_v,T'[C']_w)\) and a generating set of this coset can easily be computed
from the generating sets of the individual cosets. (Using standard techniques involving the Schreier-Sims algorithm it can be avoided that these generating sets become superpolynomially large, see e.g.~\cite{DBLP:conf/stoc/Babai16}.)

Thus, it suffices to fix one vertex \(v_C\) in each chunk \(C\) of \(T\) and then compute the isomorphism coset \(\Iso(T[C]_{v_C},T'[C']_w)\)
for each choice of \(w\in V(T')\) and its respective chunk \(C'\) via \Cref{lem:chunk_isomorphism_individualized}.
If there are \(c\) chunks in each of the patched tournaments, this leads to \(c\cdot n\) calls of the procedure from \Cref{lem:chunk_isomorphism_individualized}.
Each of these calls in turn reduces to less than \(|C|^2=\left(\frac{n}{c}\right)^2\) instances of the patched tournament isomorphism problem, each of size less than \(n^{1-\epsilon_d}\). In total, we thus reduced to less than \(\frac{n^3}{c}\leq n^3\) subinstances.
\end{proof}

\subsection{Permuting chunks}\label{sec:permuting_chunks}
Now that we have handled computing isomorphisms between chunks and automorphisms of a single chunk, it remains to consider how automorphisms may permute chunks as a collection.
For this, recall that by \Cref{lem:chunk_decomposition}, each chunk is a union of \(k\)-near-twin components, and consists of an odd number of patches.
We now want to define from \(T\) a new patched tournament \(T/{\sim^k}\) of size less than \(n^{1-\epsilon_d}\) which captures the automorphism structure on the set of chunks.

For this, let \(X\subseteq V(T)\) be a set of vertices which contains exactly one vertex from every \(k\)-near-twin component,
and let \(\mathcal{C}\) be the partition of \(X\) induced by the chunk partition.
Then the digraph \(T[X]-\mathcal{C}\) is independent of the choice of representatives in \(X\)
since any two vertices in the same \(k\)-near-twin component are twins with respect to the vertices outside of their chunks by \Cref{lem:properties_of_closure}\ref{lempart:twin_classes} (note that we need to drop arc-colors here, which is fine since preservation of arc-colors can later always be achieved by applications of the set-transporter problem).
Furthermore, since every chunk consists of an odd number of patches, the tournament \(Q(T)\) contains an odd number of 
arcs between any two distinct chunks. Thus, we can define the tournament \(Q(T/{\sim^k})\) on vertex set \(\mathcal{C}\)
by declaring \(CC'\in E(Q(T/{\sim^k}))\) if there are more \(C\)-\(C'\)-arcs than there are \(C'\)-\(C\)-arcs in \(Q(T)\), and declaring~\(C'C\in E(Q(T/{\sim^k}))\) otherwise.

Further, since we can compute isomorphism cosets \(\Iso(T[C],T[C'])\) between any two chunks via \Cref{cor:chunk_isomorphism},
we can also compute the cosets \(\Iso(T[C],T[C'])^{\sim^k_T}\leq\Iso(T[C],T[C'])\) of only those isomorphisms which preserve the
\(k\)-near twin relation \(\sim^k_T\) and thus also the partition of \(C\) and \(C'\) into \(k\)-near twin components.
This also allows us to compute the cosets \(\Iso(T[C],T[C'])/{\sim^k}\) that these cosets induce on the sets of \(k\)-near-twin components
within the chunks. This gives us a groupoid structure on the set of chunks \(\mathcal{C}\),
and by similarly computing the sets \(\Iso(T[C],T'[C'])/{\sim^k}\), this groupoid structure
extends to a groupoid structure on the union of the sets of chunks on two compatible patched tournaments.

In total, we can thus define a new patched tournament
\[T/{\sim^k}\coloneqq(T[X]-\mathcal{C},\mathcal{C},\left(\Iso(T[C],T[C'])/{\sim^k}\right)_{C,C'\in\mathcal{C}},Q(T/{\sim^k}))\]
such that for every isomorphism \(\phi\colon T\to T'\), the bijection \(\phi/{\sim^k}\) that \(\phi\)
induces on the set of \(k\)-near-twin components of \(T\) and \(T'\) induces an isomorphism
\(T/{\sim^k}\to T'/{\sim^k}\).

By \Cref{cor:patched:near-twin-structure}, this tournament has order less than \(n^{1-\epsilon_d}\)
and since we obtained \(T/{\sim^k}\) by patching the tournament \(T[X]\), this patched tournament again
has patched VC dimension at most \(d\).

\begin{lemma}\label{lem:decomposition:quotient_tournament}
Let \(T\) and \(T'\) be compatible arc-colored patched tournaments
such that the set of chunks forms vertex partitions of \(T\) and \(T'\) respectively
and such that every chunk consists of an odd number of patches.
Further, assume that for all chunks \(C\) of \(T\) and \(C'\) of \(T'\), we are given a generating set for
\(\Iso(T[C],T[C'])\).

Then in polynomial time, we can either decide that \(T\not\simeqq T'\), or we can compute
the compatible quotient tournaments \(T/{\sim^k}\) and \(T'/{\sim^k}\) and further, the computation of
\(\Iso(T,T')\) reduces in polynomial time to the computation of \(\Iso(T/{\sim^k}, T'/{\sim^k})\).
\end{lemma}
\begin{proof}
If for some chunk \(C\) in \(T\) and all chunks \(C'\) in \(T'\) the set \(\Iso(T[C],T'[C'])\) is empty,
then \(T\) and \(T'\) are not isomorphic. Thus, we may assume that for every chunk \(C\) in \(T\),
at least one of the sets \(\Iso(T[C],T'[C'])\) is non-empty, and similarly, we may assume that for every
chunk \(C'\) of \(T'\), at least one of the sets \(\Iso(T[C],T'[C'])\) is non-empty.

Thus, for any two chunks \(C_1\) and \(C_2\) in \(T\), we know that \(T[C_1]\) and \(T[C_2]\) are isomorphic
if and only if there is a chunk \(C'\) of \(T'\) such that \(T[C_1]\simeqq T'[C']\simeqq T[C_2]\) and in this case,
\(\Iso(T[C_1],T[C_2])=\phi^{-1}\circ\Iso(T[C_1],T'[C'])\) for any isomorphism \(\phi\in\Iso(T[C_2],T'[C'])\).
Similarly, we can also compute the set \(\Iso(T'[C_1'],T'[C_2'])\) for any two chunks \(C_1'\) and \(C_2'\) of \(T'\).
By restricting each of these cosets to only those bijections that preserve the partitions into \(k\)-near twin components,
and then computing the action of these cosets on the set of \(k\)-near-twin components, we can thus compute the two patched tournaments
\(T/{\sim^k}\) and \(T'/{\sim^k}\). Now, since every isomorphism \(\phi\colon T/{\sim^k}\to T'/{\sim^k}\) also induces a bijection
\(\mathcal{C}\to\mathcal{C}'\) on the set of chunks, we satisfy all the requirements of \Cref{lem:patched:iso_via_block_system} and can compute \(\Iso(T,T')\) in polynomial time.
\end{proof}

\subsection{The overall algorithm}
We are now ready to state our overall algorithm to compute isomorphism between patched tournaments. If the instances are sufficiently small, the algorithm uses the quasipolynomial algorithm from \Cref{lem:patched:patched_tournament_iso_n^logn}. Otherwise, the algorithm applies various tests from
\Cref{cor:patched:near-twin-structure,lem:chunk_decomposition} to find an invariant vertex-coloring.
If none of these tests succeed, it will have computed a partition into few, comparatively small~\(k\)-near-twin components yielding a suitable isomorphism-invariant partition into chunks. These chunks are then used as new, larger patches and the algorithm recurses.
The details are described in \Cref{algo:full_algorithm}.
\begin{algorithm} 
	\caption{An algorithm computing the set of isomorphisms between two patched tournaments of bounded patched VC dimension.}\label{algo:full_algorithm}
	\KwIn{A compatible pair of colored patched tournaments \(T\) and \(T'\) and a~$d\in \mathbb{N}$ such that \(\pVCdim^+(T)\leq d\)}
	\KwOut{The coset \(\Iso(T,T')\)}
	\If{\(|V(T)|\neq |V(T')|\)}{
		\Return \(\emptyset\)\\
	}
	Determine \(\epsilon_d>0\) and \(N_d\in\N\) from \Cref{cor:patched:near-twin-structure} and set \(k\coloneqq\left\lceil\frac{|V(T)|^{1-\epsilon_d}}{4}\right\rceil\).\\
	\If{\(|V(T)|\leq N_d\)}{
		Compute \(\Iso(T,T')\) via \Cref{lem:patched:patched_tournament_iso_n^logn}\\
		\Return \(\Iso(T,T')\)		
		\nllabel{line:recursion_base}
	}
	\ElseIf{\(T\) and \(T'\) have a nontrivial vertex-coloring}{
		Recursively compute the sets of isomorphisms \(\Iso(T[\chi^{-1}(c)],T'[(\chi')^{-1}(c)])\) for each color \(c\)\\
		Compute \(\Iso(T,T')\) using \Cref{cor:patched:nontrivial_coloring}\\
		\Return \(\Iso(T,T')\)
		\nllabel{line:nontrivial_coloring}
	}	
	Compute 	the \(k\)-near-twin relations \(\sim_T^k\) and \(\sim_{T'}^k\)\\
	\If{some vertex in either \(T\) or \(T'\) has at least \(n^{1-\epsilon_d}\) \(k\)-near twins among its neighbors}{
		Color \(V(T)\) and \(V(T')\) according to the number of in-neighbors that are \(k\)-near twins\\
		Recursively compute the sets of isomorphisms \(\Iso(T[\chi^{-1}(c)],T'[(\chi')^{-1}(c)])\) for each color \(c\)\\
		Compute \(\Iso(T,T')\) using \Cref{cor:patched:nontrivial_coloring}\\
		\Return \(\Iso(T,T')\) 
		\nllabel{line:nontrivial_coloring:degree}
	}
	For each vertex \(v\), compute the chunk \(\{v\}^{(\infty)}\)\\
	\If{not all chunks have the same size}{
		Color \(V(T)\) and \(V(T')\) according to the size of the chunks \(\{v\}^{(\infty)}\)\\
		Recursively compute the sets of isomorphisms \(\Iso(T[\chi^{-1}(c)],T'[(\chi')^{-1}(c)])\) for each color \(c\)\\
		Compute \(\Iso(T,T')\) using \Cref{cor:patched:nontrivial_coloring}\\
		\Return \(\Iso(T,T')\) 
		\nllabel{line:nontrivial_coloring:chunk_size}
	}
	\ElseIf{not all chunks contain an odd number of patches}{
		Color \(V(T)\) and \(V(T')\) according to the degree of the patches within the tournament induced by \(Q(T)\) and \(Q(T')\) on the chunks,\\
		Recursively compute the sets of isomorphisms \(\Iso(T[\chi^{-1}(c)],T'[(\chi')^{-1}(c)])\) for each color \(c\)\\
		Compute \(\Iso(T,T')\) using \Cref{cor:patched:nontrivial_coloring}\\
		\Return{\(\Iso(T,T')\)}
		\nllabel{line:nontrivial_coloring:chunks_odd}
	}
	\ElseIf{some vertex in either \(T\) or \(T'\) in some chunk \(C\) has at least \(n^{1-\epsilon_d}\) many \(\sim^k_{T[C]}\)-neighbors outside its patch in \(T[C]\) or \(T'[C]\) respectively}{
		Color \(V(T)\) and \(V(T')\) according to the number of in-neighbors in \(T[C]\) that are \(\sim^k_{T[C]}\)-neighbors,\\
		Recursively compute the sets of isomorphisms \(\Iso(T[\chi^{-1}(c)],T'[(\chi')^{-1}(c)])\) for each color \(c\)\\
		Compute \(\Iso(T,T')\) using \Cref{cor:patched:nontrivial_coloring}\\
		\Return \(\Iso(T,T')\)
		\nllabel{line:nontrivial_coloring:degree_within_chunk}
	}
	\Else(\tcp*[h]{We are now in case (b) of both \Cref{cor:patched:near-twin-structure} and \Cref{lem:chunk_decomposition}}){
		For every pair of chunks \(\{v\}^{(\infty)}\subseteq V(T)\) and \(\{w\}^{(\infty)}\subseteq V(T')\), compute a generating set for
		\(\Iso(T[\{v\}^{(\infty)}],T'[\{w\}^{(\infty)}])\) recursively via \Cref{cor:chunk_isomorphism}
		\nllabel{line:chunk_isomorphism}\\
		Compute the quotient tournaments \(T/{\sim^k}\) and \(T'/{\sim^k}\) using \Cref{lem:decomposition:quotient_tournament}\\
		Recursively compute the set \(\Iso(T/{\sim^k},T'/{\sim^k})\).
		\nllabel{line:quotient_isomorphism}\\
		Combine the isomorphism cosets between the chunks with the isomorphism coset \(\Iso(T/{\sim^k},T'/{\sim^k})\)
		and compute \(\Iso(T,T')\) via \Cref{lem:decomposition:quotient_tournament}\\
		\Return \(\Iso(T,T')\)
	}
\end{algorithm}

In total, we have found that given two compatible patched tournaments \(T\) and \(T'\) obtained by patching two tournaments of VC dimension at most \(d\),
the patched tournament isomorphism problem for \(T\) and \(T'\) reduces in polynomial time to at most \(n^3\) instances of the patched tournament isomorphism
for smaller tournaments of this kind, each of size at most \(O_d(k)\). This yields our main theorem:

\begin{theorem}\label{thm:patched_tournament_isomorphism}
Given two compatible colored patched tournaments \(T\) and \(T'\) with \(\pVCdim^+(T)\leq d\)
and \(|V(T)|=|V(T')|=n\), \Cref{algo:full_algorithm} computes a generating set of \(\Iso(T,T')\) in time \(2^{O(d^4)}n^{O(d)}\).
\end{theorem}
\begin{proof}
\emph{(Correctness)} We first prove correctness of the algorithm.
If both tournaments have order at most \(N_d\), the algorithm uses \Cref{lem:patched:patched_tournament_iso_n^logn}
to compute the isomorphism coset and we are done in this case. Thus, we may assume from now on that both tournaments have
order \(n\geq N_d\), which allows us to apply \Cref{cor:patched:near-twin-structure}. Indeed, \Cref{cor:patched:near-twin-structure}
guarantees that, either \(T\) decomposes into fewer than \(n^{1-\epsilon_d}\) many \(k\)-near-twin components and every
vertex \(v\) has fewer than \(n^{1-\epsilon_d}\) many \(k\)-near twins among its neighbors, or that coloring each vertex by the number of \(k\)-near twins among its in-neighbors yields a nontrivial, isomorphism-invariant vertex-coloring. If \(T'\) does not similarly decompose into fewer than \(n^{1-\epsilon_d}\) many \(k\)-near twin components, then the two quotient tournaments will have different orders, in which case our algorithm will return
\(\Iso(T,T')=\emptyset\).

Thus, we can either directly recurse in \Cref{line:nontrivial_coloring:degree} and combine the results using \Cref{cor:patched:nontrivial_coloring},
or may assume that both patched tournaments now satisfy the required degree bound, and also decompose into fewer than \(n^{1-\epsilon_d}\) many \(k\)-near-twin components.
Now, the algorithm computes the chunks. If not all chunks have the same size, then coloring each vertex \(v\) according to the size of the chunk \(\{v\}^{(\infty)}\)
yields a nontrivial, isomorphism-invariant coloring, which allows us to recurse on the color classes and combine using \Cref{cor:patched:nontrivial_coloring}.
Similarly, if not all chunks contain an odd number of patches or the degree bounds on the \(k\)-near-twin relation does not hold within every chunk, this again yields a nontrivial coloring on either of the two patched tournaments,
which allows us to recurse again.

Thus, we may now assume that each chunk has the same size, consists of an odd number of patches, and within each chunk, each vertex has fewer than \(n^{1-\epsilon_d}\) many \(k\)-near twins outside its patch. Via the proof of \Cref{lem:chunk_decomposition},
this implies in particular that the set of chunks now forms a partition of both patched tournaments (since otherwise, we would find a strictly smaller chunk
within any nontrivial intersection of chunks). Now, we can use \Cref{cor:chunk_isomorphism} to compute the coset \(\Iso(T[\{v\}^{(\infty)}],T'[\{w\}^{(\infty)}])\)
for every pair of chunks using fewer than \(n^3\) subinstances, each of size less than \(n^{1-\epsilon_d}\), and using these cosets, we can define
the patch structure on the quotient tournaments \(T/{\sim^k_T}\) and \(T'/{\sim^k_{T'}}\). 

\Cref{lem:decomposition:quotient_tournament} thus finally allows us to compute \(\Iso(T,T')\).

\emph{(Running time)} Now, we turn to the runtime bound.
For each real \(x\in[1,\infty)\), let \(R(x)\) be number of nodes in the recursion tree of \Cref{algo:full_algorithm}
on instances of order at most \(\lfloor x\rfloor\).
Since all of the steps performed in the algorithm apart from the recursive calls run in polynomial time for all instances of order at least \(N_d\)
and in time at most \(N_d^{O(\log N_d)}\in 2^{O(d^4)}\) on instances of order at most \(N_d\),
the running time of \Cref{algo:full_algorithm} on instances of order \(n\) is then bounded by \(2^{O(d^4)}n^{O(1)}\cdot R(n)\).
It thus remains to bound the function \(R(n)\). Clearly, \(R(n)=1\) for all \(n\leq N_d\).

If \(n> N_d\), then \Cref{algo:full_algorithm} reduces to recursive subinstances in two different cases. Either, we have found a nontrivial vertex coloring
and recurse on each pair of corresponding color classes (this happens in \Cref{line:nontrivial_coloring,line:nontrivial_coloring:degree,line:nontrivial_coloring:chunk_size,line:nontrivial_coloring:chunks_odd,line:nontrivial_coloring:degree_within_chunk}),
or the set of chunks forms a suitable partition and we compute isomorphism cosets between each pair of chunks in \Cref{line:chunk_isomorphism}
via \Cref{cor:chunk_isomorphism}
and between the quotient tournaments \(T/{\sim^k_T}\) and \(T'/{\sim^k_{T'}}\) in \Cref{line:quotient_isomorphism}.
In the former case, we thus reduce to instances of size \(n_1,\dots,n_\ell\) with \(\sum_{i=1}^\ell n_i=n\),
while in the second case, \Cref{cor:chunk_isomorphism} implies that we have at most \(n^3\) subinstances, each of size less than \(n^{1-\epsilon_d}\).

This yields the recurrence inequality
\[R(n)\leq \max\left\{n^3R(n^{1-\epsilon_d}),\max_{n=\sum n_i}\sum R(n_i)\right\}.\]
We claim that this inequality implies that \(R(n)\leq n^{\frac{3}{\epsilon_d}}\).
Indeed, this is true for every \(n\leq N_d\). Now, assume this is true for all inputs less than \(n\).
Then
\begin{align*}
R(n)&\leq \max\left\{n^3R(n^{1-\epsilon_d}),\max_{n=\sum n_i}\sum R(n_i)\right\}\\
&\leq \max\left\{n^3\left(n^{1-\epsilon_d}\right)^{\frac{3}{\epsilon_d}}, \max_{n=\sum_{i=1}^\ell n_i}\sum_{i=1}^\ell n_i^{\frac{3}{\epsilon_d}}\right\}\\
&\leq\max\left\{n^3\cdot n^{\frac{3}{\epsilon_d}-3},\max_{n=\sum_{i=1}^\ell n_i}\left(\sum_{i=1}^\ell n_i\right)^{\frac{3}{\epsilon_d}}\right\}\\
&=\max\left\{n^{\frac{3}{\epsilon_d}},n^{\frac{3}{\epsilon_d}}\right\}\\
&=n^{\frac{3}{\epsilon_d}}.
\end{align*}

Thus, we have shown by induction that \(R(n)\leq n^{\frac{3}{\epsilon_d}}\) for all \(n\in\N\). In total, the algorithm thus runs in time
\(2^{O(d^4)}\cdot n^{O(1)}\cdot n^{\frac{3}{\epsilon_d}}=2^{O(d^4)}\cdot n^{O(d)}\).
\end{proof}

Finally, we can apply our algorithm directly to tournaments of bounded VC dimension.
\begin{theorem}\label{thm:tournament_isomorphism}
Given two (vertex- and arc-colored) tournaments \(T\) and \(T'\), we can compute \(\Iso(T,T')\) in time
\(2^{O(d^4(\log d)^4)}n^{O(d\log d)}\) where \(d=\min\{\VCdim^+(T),\VCdim^+(T')\}\).
\end{theorem}
\begin{proof}
In a first step, we can determine \(d\) in time \(n^{O(d)}\) by simply iterating through all sets of increasing size until we find some size
of which there is no shattered set in one of the two tournaments. If \(\VCdim^+(T)\neq\VCdim^+(T')\), we can immediately decide \(T\not\simeqq T'\).
Otherwise, we can view \(T\) and \(T'\) as compatible patched tournaments with the discrete patch partition.
By \Cref{lem:patched:VC-dim_bound}, both tournaments have patched VC dimension at most \(O(d\log d)\).
Thus, \Cref{thm:patched_tournament_isomorphism} implies that \Cref{algo:full_algorithm} determines \(\Iso(T,T')\) in time \(2^{O(d^4(\log d)^4)}n^{O(d\log d)}\).
\end{proof}

Before we go on to the tournament isomorphism for tournaments of bounded chromatic number, let us just mention that
classes of bounded VC dimension are special in that they are precisely those hereditary classes that contain
at most \(2^{O(n^{2-\epsilon})}\) labeled \(n\)-vertex tournaments for some \(\epsilon>0\), see \Cref{sec:VC:growth}. 
\section{Isomorphism of \texorpdfstring{\(k\)}{k}-colorable tournaments}\label{sec:k:colorable:tournaments}
A tournament \(T\) is called \emph{\(k\)-colorable} if there exists a partition \(V(T)=V_1\dotcup\dots\dotcup V_k\) such that each of the tournaments \(T[V_i]\) is transitive, i.e., has a transitive arc relation. The minimal \(k\) such that~$T$ is~$k$-colorable is called the \emph{chromatic number} of \(T\) and is denoted by \(\chi(T)\).
 
The \(2\)-colorable tournaments play an important role in the theory of VC dimension since a hereditary class of tournaments
has bounded VC dimension if and only if it does not contain every \(2\)-colorable tournament~\cite{domination_tournaments}.
Thus, the class of \(2\)-colorable tournaments is the unique minimal hereditary class of unbounded VC dimension.
In this section, we prove that isomorphism of \(k\)-colorable tournaments can nonetheless be decided in polynomial time \(n^{O(k)}\).

We start with isomorphism of \(2\)-colorable tournaments, because even though we will not need this as a subcase in the general algorithm,
it already contains the main ingredient of choosing vertices whose in- and out-neighborhoods have smaller chromatic number.

In the \(2\)-colorable case, the underlying structural lemma for this is the following:
\begin{lemma}\label{lem:chromatic:2col_structure}
Let \(T\) be a regular tournament and \(V(T)=A\dotcup B\) a \(2\)-coloring with \(|A|\leq|B|\).
Then there exists a vertex \(v\in V(T)\) such that \(A=N^-(v)\) and \(B=N^+[v]\).
\end{lemma}
\begin{proof}
Because \(T\) is regular, it has odd order and thus \(|A|\leq\frac{n-1}{2}\) and \(|B|\geq\frac{n+1}{2}\).
Now, let \(v\) be the unique vertex of in-degree \(0\) in the transitive tournament \(T[B]\).
The vertex \(v\) thus has \(|B|-1\geq \frac{n-1}{2}\) out-neighbors within \(B\), and since
every vertex in a regular tournament has out-degree exactly \(\frac{n-1}{2}\), \(B=N^+[v]\)
and all vertices in \(A\) are its in-neighbors. Thus, the vertex \(v\) fulfills the requirements of the claim.
\end{proof}
The lemma above implies that to check whether a regular tournament is \(2\)-colorable,
it suffices to check at most \(n\) candidate \(2\)-colorings. In particular, this problem is
decidable in polynomial time. Note that this is in contrast to the case of arbitrary tournaments,
where deciding \(2\)-colorability is NP-hard~\cite{tournament_2col_NP-hard}.

\begin{theorem}\label{thm:chromatic:2col_isomorphism}
There is an isomorphism-invariant polynomial-time algorithm which, given two (vertex- and arc-colored) tournaments \(T\) and \(T'\),
either determines that \(T\) is not \(2\)-colorable, or computes a generating set of \(\Iso(T,T')\). 
\end{theorem}
\begin{proof}
By coloring each vertex by its out-degree, handling the color classes independently
and repeating this if necessary, we may assume that the two given tournaments \(T\) and \(T'\) are both regular.

By \Cref{lem:chromatic:2col_structure}, if \(T\) is \(2\)-colorable,
there is a vertex \(v\in V(T)\) such that \(V(T)=N^-(v)\dotcup N^+[v]\) is a \(2\)-coloring of \(T\)
and we can find such a vertex by trying out all vertices until we find one. Thus, we can in polynomial time
either find such a vertex \(v\), or determine that \(T\) is not \(2\)-colorable.

Now, it suffices to check for every vertex \(w\in V(T')\) whether there is an isomorphism \(\phi\colon T\to T'\) with \(\phi(v)=w\).
If \(N^-(w)\dotcup N^+[w]\) does not define a \(2\)-coloring of \(T'\), then there is no such isomorphism.
Otherwise, there is at most one such isomorphism: we must map \(v\) to \(w\), the set \(N^-(v)\) to \(N^-(w)\), and the set
\(N^+(v)\) to \(N^+(w)\). But since each of these sets induces a transitive subtournament, there is only one isomorphism
between each of these subtournaments. Thus, we can simply check whether this unique bijection is an isomorphism.
By doing this for all choices of \(w\), we find the set of all isomorphisms \(T\to T'\), and can easily extract a generating set from this.
\end{proof}

In order to extend this result to tournaments of arbitrary bounded chromatic number, we need another technical lemma.
We call a vertex \(v\in V(T)\) \emph{out-reducing} if \(\chi(T[N^+(v)])<\chi(T)\) and \emph{in-reducing} if \(\chi(T[N^-(v)])<\chi(T)\).
Note that in these terms, \Cref{lem:chromatic:2col_structure} essentially stated that every regular \(2\)-colorable tournament
contains some vertex which is both in- and out-reducing, which allows us to recurse on the in- and out-neighborhood of this vertex independently.
While this is no longer true for tournaments of larger chromatic number, we at least still get in- and out-reducing vertices individually.
\begin{lemma}\label{lem:chromatic:reducing}
Every non-transitive tournament contains both an out-reducing vertex and an in-reducing vertex.
\end{lemma}
\begin{proof}
Let \(V(T)=V_1\dotcup\dots\dotcup V_{\chi(T)}\) be a \(\chi(T)\)-coloring of \(T\).
Let \(v^+\) and \(v^-\) be the vertices of
out- and in-degree \(0\) in the transitive tournament \(T[V_1]\). Thus, \(V_1\subseteq N^+[v^-]\) and \(V_1\subseteq N^-[v^+]\).
In particular, the two tournaments \(T[N^+(v^+)]\) and \(T[N^-(v^-)]\) do not contain vertices from \(V_1\) and are thus \((\chi(T)-1)\)-colorable.
Thus, the vertex \(v^+\) is out-reducing and the vertex \(v^-\) is in-reducing.
\end{proof}

With this notion at hand, we state the main idea of our algorithm.
Given two \(k\)-colorable tournaments \(T\) and \(T'\), and an arbitrary vertex~$v\in T$, 
the computation of \(\Iso(T,T')\) reduces to computing \(\Iso(T[N^+(v)],T'[N^+(w)])\) and \(\Iso(T[N^-(v)],T'[N^-(w)])\)
for all vertices \(w\in V(T')\).
But note that if the vertex \(v\) is both in- and out-reducing,
then both~$T[N^+(v)]$ and~$T[N^-(v)]$ are \((k-1)\)-colorable tournaments, and thus we can solve these instances recursively.
If \(T\), however, does not contain such a vertex \(v\) which is both in- and out-reducing, we get a nontrivial vertex-coloring,
since we know that \(T\) contains both in- and out-reducing vertices.
The main complication with this approach is that we do not know how to compute this coloring in polynomial time
since computing the chromatic number of the tournaments \(T[N^+(v)]\) and \(T[N^-(v)]\) is NP-hard~\cite{tournament_2col_NP-hard,tournament_kcol_NP-hard}.
Instead, we use our isomorphism algorithm itself:
In order for a vertex \(v\) to be a suitable replacement for an out-reducing vertex, it suffices that
all of the computations of \(\Iso(T[N^+(v)],T'[N^+(w)])\) with \(w\in V(T')\)
via the algorithm for chromatic number \(k-1\) succeed. In order for this to work,
our algorithm may return that \(\chi(T)>k\) if it has determined that this is the case,
but it may never return a wrong set \(\Iso(T,T')\) even when \(T\) and \(T'\) are not \(k\)-colorable.

\begin{lemma}\label{lem:chromatic:kcol_isomorphism_technical}
There is an isomorphism-invariant algorithm which, given two vertex- and arc-colored tournaments \(T\) and \(T'\) of order at most \(n\) and some \(k\in\N\), runs in time \(n^{O(k)}\)
and either
\begin{enumerate}
\item correctly determines that \(\chi(T)>k\), or
\item computes a generating set for \(\Iso(T,T')\).
\end{enumerate}
Furthermore, whether the algorithm returns that \(\chi(T)>k\) or a generating set of \(\Iso(T,T')\) depends only on the isomorphism types of \(T\) and \(T'\).
\end{lemma}
\begin{proof}
We show the statement by induction on~$n$ and $k$.

If \(n\leq 1\), we can directly compute the set \(\Iso(T,T')\) in linear time.
If instead \(k=1\), we can check whether the given tournaments are transitive and if so, compute the set \(\Iso(T,T')\) in linear time by checking whether the unique arc-preserving map is color-preserving.

From now on, we may thus assume that we have already constructed an algorithm that works for all tournaments of orders smaller than \(n\) and for all values smaller than \(k\).
We write \(\Pi(S,S',k')\) for the output of this algorithm on input \(S\), \(S'\) with either the same parameter \(k'=k\) and orders smaller than \(n\), or parameter \(k'<k\).
Now, for all pairs of vertices \(v\in V(T)\) and \(w\in V(T')\), we recursively compute
the four subinstances
\begin{enumerate}
\item \(\Pi(T[N^+(v)],T'[N^+(w)],k-1)\),
\item \(\Pi(T[N^-(v)],T'[N^-(w)],k-1)\),
\item \(\Pi(T'[N^+(w)],T[N^+(v)],k-1)\), and
\item \(\Pi(T'[N^-(w)],T[N^-(v)],k-1)\).
\end{enumerate}

We now call a vertex \(v\in V(T)\)
\begin{enumerate}
\item \emph{quasi-out-reducing} if none of the computations \(\Pi(T[N^+(v)],T'[N^+(w)],k-1)\) with \(w\in V(T')\) return that \(\chi(T[N^+(v)])>k-1\) and
\item \emph{quasi-in-reducing}  if none of the computations \(\Pi(T[N^-(v)],T'[N^-(w)],k-1)\) with \(w\in V(T')\) return that \(\chi(T[N^-(v)])>k-1\).
\end{enumerate}
Similarly, we call a vertex \(w\in V(T')\) \emph{quasi-out-reducing} if for no vertex \(v\in V(T)\), the computation \(\Pi(T'[N^+(w)],T[N^+(v)],k-1)\) returns that \(\chi(T'[N^+(w)])> k-1\)
and \emph{quasi-in-reducing} if none of the computations of \(\Pi(T'[N^-(w)],T[N^-(v)],k-1)\) with \(v\in V(T)\) returned that \(\chi(T'[N^-(w)])> k-1\).
\begin{claim}
If \(T\) is \(k\)-colorable, then every out-reducing vertex in \(T\) is quasi-out-reducing and every in-reducing vertex in \(T\) is quasi-in-reducing.
\end{claim}
\begin{claimproof}
If \(v\in V(T)\) is out-reducing, this means that \(\chi(T[N^+(v)])<\chi(T)\leq k\).
Thus, the respective recursive subinstances \(\Pi(T[N^+(v)],T'[N^+(w)],k-1)\) all satisfy \(\chi(T[N^+(v)])\leq k-1\)
and thus cannot return \(\chi(T[N^+(v)])>k-1\) by correctness of the algorithm. Thus, \(v\) is quasi-out-reducing.

The proof that in-reducing vertices are quasi-in-reducing, and the case of vertices \(w\in V(T')\) are completely analogous.
\end{claimproof}
Thus, by \Cref{lem:chromatic:reducing}, if \(T\) does not contain both a quasi-out-reducing vertex and a quasi-in-reducing vertex,
it must either be transitive (in which case we can easily compute the set \(\Iso(T,T')\)), or we can immediately return that \(\chi(T)>k\).
If instead \(T'\) does not contain both a quasi-out-reducing vertex and a quasi-in-reducing vertex but \(T\) does, we can return \(\Iso(T,T')=\emptyset\) since the sets of quasi-out- and quasi-in-reducing vertices are preserved by all isomorphisms. Indeed, if \(\phi\colon T\to T'\) is an isomorphism
and \(v\) is quasi-out-reducing, then \(\Pi(T[N^+(v)],T'[N^+(w)],k-1)\) must always give the same answer as \(\Pi(T'[N^+(\phi(v))],T[N^+(\phi^{-1}(w))],k-1)\), which means that also \(\phi(v)\) is quasi-out-reducing. The analogous argument shows that also the set of quasi-in-reducing vertices
is preserved by every isomorphism.

\begin{claim}
If some vertex \(v\in V(T)\) is both quasi-in- and quasi-out-reducing, then we can compute the set \(\Iso(T,T')\) in polynomial time.
\end{claim}
\begin{claimproof}
If \(v\in V(T)\) is both quasi-in- and quasi-out-reducing, then we have already
successfully computed both \(\Iso(T[N^+(v)],T'[N^+(w)])\) and \(\Iso(T[N^-(v)],T'[N^-(w)])\) for all vertices \(w\in V(T')\).
We can combine these sets to find the set of isomorphisms \(\Iso(T_v,T'_w)\) consisting
of those isomorphisms \(\phi\colon T\to T'\) for which \(\phi(v)=w\). If all of these sets are empty, then so is \(\Iso(T,T')\).
Otherwise, we combine the cosets using the standard technique as follows: we write \(\Iso(T_v,T'_w)=\phi_w\circ\langle S\rangle\) for each \(w\in V(T')\) for which \(\Iso(T_v,T'_w)\) is nonempty,
where \(S\) is a generating set of \(\Stab_{\Aut(T)}(v)\).
Then, we fix one such \(w\) and get
\[\Iso(T,T')=\phi_w\circ\left\langle S\cup \{\phi_{w'}^{-1}\circ\phi_w\colon \Iso(T_v,T'_{w'})\neq\emptyset\}\right\rangle.\qedhere\]
\end{claimproof}
Thus, either we are already done by the previous claim, or find a nontrivial isomorphism-invariant vertex-coloring on both \(T\) and \(T'\) by coloring
each vertex according to whether it is quasi-in-reducing, quasi-out-reducing, both, or neither.
Thus, we can recursively call \(\Pi(T[C],T'[C'],k)\) on each pair of corresponding color classes \(C\) and \(C'\) and combine the isomorphism cosets afterwards.

The full algorithm can be found in \Cref{algo:chromatic}.
\begin{algorithm} 
	\caption{An algorithm computing the set of isomorphisms between two \(k\)-colorable tournaments.}\label{algo:chromatic}
	\KwIn{Colored tournaments \(T\) and \(T'\) and some \(k\in\N\)}
	\KwOut{Either the information that \(\chi(T)>k\), or a generating set of \(\Iso(T,T')\)}
	\If{\(|V(T)|\neq |V(T')|\)}{
		\Return \(\Iso(T,T')=\emptyset\)
	}
	\ElseIf{\(|V(T)|\leq 1\)}{
		Check whether the unique bijection \(V(T)\to V(T')\) is vertex-color-preserving\\
		\Return \(\Iso(T,T')\).
	}
	\ElseIf{\(k=1\)}{
		\If{\(T\) is transitive}{
			\If{\(T'\) is transitive} {
				Compute \(\Iso(T,T')\) by checking whether the unique arc-preserving map is color-preserving.\\
				\Return \(\Iso(T,T')\).
			}
			\Else{
				\Return \(\Iso(T,T')=\emptyset\).
			}
		}
		\Else{
			\Return \(\chi(T)>k\).		
		}
	}
	\Else{
		\ForEach{\(v\in V(T)\) and \(w\in V(T')\)}{
			\nllabel{line:chromatic:recursive_k-1}
			Recursively compute \(\Iso(T[N^+(v)],T'[N^+(w)])\), \(\Iso(T[N^-(v)],T'[N^-(w)])\), \(\Iso(T'[N^+(w)],T[N^+(v)])\), and \(\Iso(T'[N^-(w)],T[N^-(v)])\) with chromatic number bound \(k-1\)\\
		}
		Label a vertex \(v\in V(T)\) as \emph{quasi out-reducing} if for no~$w$ the computation of \(\Iso(T[N^+(v)],T'[N^+(w)])\) returns \(\chi(T[N^+(v)])>k-1\),
		and \emph{quasi in-reducing} if for no~$w$ the computation of \(\Iso(T[N^-(v)],T'[N^-(w)])\) returns \(\chi(T[N^-(v)])>k-1\)\\
		Similarly define labels in \(T'\).\\
		\If{some vertex \(v\in V(T)\) is both quasi-in and quasi-out-reducing}{
			For every vertex \(w\in V(T')\), compute \(\Iso(T_v,T'_w)\) from
			the precomputed isomorphism sets \(\Iso(T[N^+(v)],T'[N^+(w)])\) and \(\Iso(T[N^-(v)],T'[N^-(w)])\).\\
			Combine these cosets into a common coset \(\Iso(T,T')\).\\
			\Return \(\Iso(T,T')\).\\
		}
		\ElseIf{there are no quasi-in- or no quasi-out-reducing vertices in \(T\)}{
			\Return \(\chi(T)>k\)\\
		}
		\Else{
			\nllabel{line:chromatic:recursive_k}Recurse on each label class and combine isomorphism sets.\\
			\Return \(\Iso(T,T')\)
		}
	}
\end{algorithm}

\emph{(Running time)}
Regarding the running time, we first argue that if we run the algorithm on two tournaments \(T\) and \(T'\) of order \(n\) with parameter \(k\),
then there are at most \(O(n)\) many recursive subinstances with the same parameter \(k\). Indeed, note that the only case in which
we do not decrease the parameter \(k\) is when we have found a nontrivial vertex-coloring of \(T\) and \(T'\) and recurse on each pair
of corresponding color classes independently in \Cref{line:chromatic:recursive_k}.
Thus, the leaves of this recursion tree correspond to a partition of \(V(T)\) and \(V(T')\) respectively.
Since further, in this tree each non-leaf node has at least two children, the number of total nodes is linearly bounded in the number of leaves,
which in turn is linearly bounded in \(n\).

Next, consider the number of recursive subinstances with parameter \(k-1\) which are called directly from some instance with parameter \(k\).
We already argued that there are only \(O(n)\) instances with parameter \(k\), and each of these instances immediately calls \(4n^2\) subinstances
with parameter \(k-1\) in \Cref{line:chromatic:recursive_k-1}. In total, the instances with parameter \(k\) thus directly call \(O(n^3)\) subinstances
with parameter \(k-1\) and order at most \(n\). Each of these in turn calls \(O(n^3)\) subinstances with parameter \(k-2\) and in total,
we get \(n^{O(k)}\) subinstances with arbitrary parameter. Since the base instances with \(n=1\) or \(k=1\) can be solved in polynomial time,
and further, also all non-recursive work takes only polynomial time, the total algorithm runs in time \(n^{O(1)}\cdot n^{O(k)}=n^{O(k)}\).
\end{proof}

\begin{theorem}
There is an algorithm which, given two tournaments \(T\) and \(T'\) of order \(n\),
computes a generating set for the coset of isomorphisms \(\Iso(T,T')\) in time \(n^{O(\chi(T))}\).
\end{theorem}
\begin{proof}
By running the algorithm from \Cref{lem:chromatic:kcol_isomorphism_technical} on the tournaments \(T\) and \(T'\) for increasing values of \(k\),
we can find the set \(\Iso(T,T')\) in time \(\chi(T)\cdot n^{O(\chi(T))}=n^{O(\chi(T))}\).
\end{proof}
 
\bibliography{bibliography.bib}

\begin{thebibliography}{10}

\bibitem{DBLP:conf/isaac/ArvindDM06}
Vikraman Arvind, Bireswar Das, and Partha Mukhopadhyay.
\newblock On isomorphism and canonization of tournaments and hypertournaments.
\newblock In {\em Algorithms and Computation, 17th International Symposium,
  {ISAAC} 2006, Kolkata, India, December 18-20, 2006, Proceedings}, volume 4288
  of {\em Lecture Notes in Computer Science}, pages 449--459. Springer, 2006.
\newblock \href {https://doi.org/10.1007/11940128\_46}
  {\path{doi:10.1007/11940128\_46}}.

\bibitem{DBLP:journals/adam/ArvindPR25}
Vikraman Arvind, Ilia Ponomarenko, and Grigory Ryabov.
\newblock Isomorphism testing of k-spanning tournaments is fixed parameter
  tractable.
\newblock {\em Art Discret. Appl. Math.}, 8(2):2, 2025.
\newblock \href {https://doi.org/10.26493/2590-9770.1712.3EC}
  {\path{doi:10.26493/2590-9770.1712.3EC}}.

\bibitem{VCdim_dual}
Patrick Assouad.
\newblock Densit\'e et dimension.
\newblock {\em Annales de l'Institut Fourier}, 33(3):233--282, 1983.
\newblock \href {https://doi.org/10.5802/aif.938} {\path{doi:10.5802/aif.938}}.

\bibitem{DBLP:conf/stoc/Babai16}
L{\'{a}}szl{\'{o}} Babai.
\newblock Graph isomorphism in quasipolynomial time [extended abstract].
\newblock In {\em Proceedings of the 48th Annual {ACM} {SIGACT} Symposium on
  Theory of Computing, {STOC} 2016, Cambridge, MA, USA, June 18-21, 2016},
  pages 684--697. {ACM}, 2016.
\newblock \href {https://doi.org/10.1145/2897518.2897542}
  {\path{doi:10.1145/2897518.2897542}}.

\bibitem{DBLP:conf/stoc/BabaiL83}
L{\'{a}}szl{\'{o}} Babai and Eugene~M. Luks.
\newblock Canonical labeling of graphs.
\newblock In {\em Proceedings of the 15th Annual {ACM} Symposium on Theory of
  Computing, 25-27 April, 1983, Boston, Massachusetts, {USA}}, pages 171--183.
  {ACM}, 1983.
\newblock \href {https://doi.org/10.1145/800061.808746}
  {\path{doi:10.1145/800061.808746}}.

\bibitem{tournament_2col_NP-hard}
Xujin Chen, Xiaodong Hu, and Wenan Zang.
\newblock A min-max theorem on tournaments.
\newblock {\em {SIAM} J. Comput.}, 37(3):923--937, 2007.
\newblock \href {https://doi.org/10.1137/060649987}
  {\path{doi:10.1137/060649987}}.

\bibitem{domination_tournaments}
Maria Chudnovsky, Ringi Kim, Chun-Hung Liu, Paul Seymour, and Stéphan
  Thomassé.
\newblock {Domination in tournaments}.
\newblock {\em Journal of Combinatorial Theory, Series B}, 130:98--113, 2018.
\newblock \href {https://doi.org/10.1016/j.jctb.2017.10.001}
  {\path{doi:10.1016/j.jctb.2017.10.001}}.

\bibitem{DBLP:journals/corr/abs-2604-12584}
Anatole Dahan, Martin Grohe, Daniel Neuen, and Tom{\'{a}}s Novotn{\'{y}}.
\newblock Robust graph isomorphism, quadratic assignment and {VC} dimension.
\newblock {\em CoRR}, abs/2604.12584, 2026.
\newblock \href {https://arxiv.org/abs/2604.12584} {\path{arXiv:2604.12584}},
  \href {https://doi.org/10.48550/ARXIV.2604.12584}
  {\path{doi:10.48550/ARXIV.2604.12584}}.

\bibitem{DBLP:journals/endm/EvdokimovP05}
Sergei Evdokimov and Ilia~N. Ponomarenko.
\newblock Circulant graphs: efficient recognizing and isomorphism testing:
  (extended abstract).
\newblock {\em Electron. Notes Discret. Math.}, 22:7--12, 2005.
\newblock \href {https://doi.org/10.1016/J.ENDM.2005.06.002}
  {\path{doi:10.1016/J.ENDM.2005.06.002}}.

\bibitem{FeitThompson}
Walter Feit and John~G. Thompson.
\newblock Solvability of groups of odd order.
\newblock {\em Pacific J. Math.}, 13:775--1029, 1963.
\newblock \href {https://doi.org/10.2140/pjm.1963.13-3}
  {\path{doi:10.2140/pjm.1963.13-3}}.

\bibitem{tournament_kcol_NP-hard}
Jacob Fox, Lior Gishboliner, Asaf Shapira, and Raphael Yuster.
\newblock The removal lemma for tournaments.
\newblock {\em J. Comb. Theory {B}}, 136:110--134, 2019.
\newblock \href {https://doi.org/10.1016/J.JCTB.2018.10.001}
  {\path{doi:10.1016/J.JCTB.2018.10.001}}.

\bibitem{tww_tournaments}
Colin Geniet and Stéphan Thomassé.
\newblock First order logic and twin-width in tournaments and dense oriented
  graphs.
\newblock {\em European Journal of Combinatorics}, 132:104247, 2026.
\newblock \href {https://doi.org/10.1016/j.ejc.2025.104247}
  {\path{doi:10.1016/j.ejc.2025.104247}}.

\bibitem{tournament_isomorphism_tww}
Martin Grohe and Daniel Neuen.
\newblock {Isomorphism for Tournaments of Small Twin Width}.
\newblock In {\em 51st International Colloquium on Automata, Languages, and
  Programming (ICALP 2024)}, volume 297 of {\em Leibniz International
  Proceedings in Informatics (LIPIcs)}, pages 78:1--78:20, Dagstuhl, Germany,
  2024. Schloss Dagstuhl -- Leibniz-Zentrum f{\"u}r Informatik.
\newblock \href {https://doi.org/10.4230/LIPIcs.ICALP.2024.78}
  {\path{doi:10.4230/LIPIcs.ICALP.2024.78}}.

\bibitem{sphere_packing}
David Haussler.
\newblock {Sphere packing numbers for subsets of the Boolean $n$-cube with
  bounded Vapnik-Chervonenkis dimension}.
\newblock {\em Journal of Combinatorial Theory, Series A}, 69(2):217--232,
  1995.
\newblock \href {https://doi.org/10.1016/0097-3165(95)90052-7}
  {\path{doi:10.1016/0097-3165(95)90052-7}}.

\bibitem{bounded_degree_isomorphism}
Eugene~M. Luks.
\newblock Isomorphism of graphs of bounded valence can be tested in polynomial
  time.
\newblock {\em Journal of Computer and System Sciences}, 25(1):42--65, 1982.
\newblock \href {https://doi.org/10.1016/0022-0000(82)90009-5}
  {\path{doi:10.1016/0022-0000(82)90009-5}}.

\bibitem{DBLP:conf/dimacs/Luks91}
Eugene~M. Luks.
\newblock Permutation groups and polynomial-time computation.
\newblock In {\em Groups And Computation, Proceedings of a {DIMACS} Workshop,
  New Brunswick, New Jersey, USA, October 7-10, 1991}, volume~11 of {\em
  {DIMACS} Series in Discrete Mathematics and Theoretical Computer Science},
  pages 139--175. {DIMACS/AMS}, 1991.
\newblock \href {https://doi.org/10.1090/DIMACS/011/11}
  {\path{doi:10.1090/DIMACS/011/11}}.

\bibitem{group_isomorphism}
Gary~L. Miller.
\newblock On the n{\^{}}log n isomorphism technique: {A} preliminary report.
\newblock In {\em Proceedings of the 10th Annual {ACM} Symposium on Theory of
  Computing, May 1-3, 1978, San Diego, California, {USA}}, pages 51--58. {ACM},
  1978.
\newblock \href {https://doi.org/10.1145/800133.804331}
  {\path{doi:10.1145/800133.804331}}.

\bibitem{hypergraph_isomorphism}
Gary~L. Miller.
\newblock Isomorphism of graphs which are pairwise k-separable.
\newblock {\em Information and Control}, 56(1):21--33, 1983.
\newblock \href {https://doi.org/10.1016/S0019-9958(83)80048-5}
  {\path{doi:10.1016/S0019-9958(83)80048-5}}.

\bibitem{DBLP:journals/jct/MuzychukP20}
Mikhail~E. Muzychuk and Ilia Ponomarenko.
\newblock Testing isomorphism of circulant objects in polynomial time.
\newblock {\em J. Comb. Theory {A}}, 169, 2020.
\newblock \href {https://doi.org/10.1016/J.JCTA.2019.105128}
  {\path{doi:10.1016/J.JCTA.2019.105128}}.

\bibitem{hypergraph_isomorphism_neuen}
Daniel Neuen.
\newblock Hypergraph isomorphism for groups with restricted composition
  factors.
\newblock {\em {ACM} Trans. Algorithms}, 18(3):27:1--27:50, 2022.
\newblock \href {https://doi.org/10.1145/3527667} {\path{doi:10.1145/3527667}}.

\bibitem{Ponomarenko1992}
I.~N. Ponomarenko.
\newblock Polynomial time algorithms for recognizing and isomorphism testing of
  cyclic tournaments.
\newblock {\em Acta Applicandae Mathematicae}, 29(1):139--160, 1992.
\newblock \href {https://doi.org/10.1007/BF00053383}
  {\path{doi:10.1007/BF00053383}}.

\bibitem{MR2981982}
I.~N. Ponomarenko.
\newblock Bases of {S}churian antisymmetric coherent configurations and
  isomorphism test for {S}churian tournaments.
\newblock {\em Zap. Nauchn. Sem. S.-Peterburg. Otdel. Mat. Inst. Steklov.
  (POMI)}, 402:108--147, 219--220, 2012.
\newblock \href {https://doi.org/10.1007/s10958-013-1398-2}
  {\path{doi:10.1007/s10958-013-1398-2}}.

\bibitem{sauer_shelah1}
Norbert Sauer.
\newblock On the density of families of sets.
\newblock {\em J. Comb. Theory {A}}, 13(1):145--147, 1972.
\newblock \href {https://doi.org/10.1016/0097-3165(72)90019-2}
  {\path{doi:10.1016/0097-3165(72)90019-2}}.

\bibitem{DBLP:conf/icalp/Schweitzer17}
Pascal Schweitzer.
\newblock A polynomial-time randomized reduction from tournament isomorphism to
  tournament asymmetry.
\newblock In {\em 44th International Colloquium on Automata, Languages, and
  Programming, {ICALP} 2017, Warsaw, Poland, July 10-14, 2017}, volume~80 of
  {\em LIPIcs}, pages 66:1--66:14. Schloss Dagstuhl - Leibniz-Zentrum f{\"{u}}r
  Informatik, 2017.
\newblock \href {https://doi.org/10.4230/LIPICS.ICALP.2017.66}
  {\path{doi:10.4230/LIPICS.ICALP.2017.66}}.

\bibitem{sauer_shelah2}
Saharon Shelah.
\newblock A combinatorial problem; stability and order for models and theories
  in infinitary languages.
\newblock {\em Pacific J. Math.}, 41:247--261, 1972.
\newblock \href {https://doi.org/10.2140/pjm.1972.41.247}
  {\path{doi:10.2140/pjm.1972.41.247}}.

\bibitem{VC}
V.~N. Vapnik and A.~Ya. Chervonenkis.
\newblock On the uniform convergence of relative frequencies of events to their
  probabilities.
\newblock {\em Theory of Probability \& Its Applications}, 16(2):264--280,
  1971.
\newblock \href {https://doi.org/10.1137/1116025} {\path{doi:10.1137/1116025}}.

\bibitem{DBLP:conf/mfcs/Wagner07}
Fabian Wagner.
\newblock Hardness results for tournament isomorphism and automorphism.
\newblock In {\em Mathematical Foundations of Computer Science 2007, 32nd
  International Symposium, {MFCS} 2007, Cesk{\'{y}} Krumlov, Czech Republic,
  August 26-31, 2007, Proceedings}, volume 4708 of {\em Lecture Notes in
  Computer Science}, pages 572--583. Springer, 2007.
\newblock \href {https://doi.org/10.1007/978-3-540-74456-6\_51}
  {\path{doi:10.1007/978-3-540-74456-6\_51}}.

\end{thebibliography}
\bibliographystyle{plainurl}

\appendix

\section{VC dimension and growth}\label{sec:VC:growth}
Since every hereditary class of tournaments of unbounded VC dimension contains all \(2\)-colorable tournaments
and there are \(2^{|A|\times|B|}\) labeled tournaments with a prescribed \(2\)-coloring \(A\dotcup B\),
every hereditary class of tournaments of unbounded VC dimension contains at least \(2^{n^2/4}\) labeled
tournaments of order \(n\). This is in contrast to classes of bounded VC dimension, which for every \(\epsilon>0\)
contain at most \(2^{\epsilon n^2}\) such tournaments for every large enough \(n\) by \cite{domination_tournaments}.
Indeed, we can even show that they contain only \(2^{n^{2-\epsilon}}\) such tournaments for some \(\epsilon>0\):
\begin{theorem}
Let \(\mathcal{C}\) be a class of tournaments of bounded VC dimension. Then there is some \(\epsilon>0\)
such that for all sufficiently large \(n\), the class \(\mathcal{C}\) contains at most \(O(2^{n^{2-\epsilon}})\) labeled tournaments
of order \(n\).
\end{theorem}
\begin{proof}
Assume that every tournament in \(\mathcal{C}\) has VC dimension at most \(d\).
By \Cref{cor:near-twins:existence}, there is some \(N_d\in\N\) and some \(\epsilon_d>0\)
such that every tournament \(T\in\mathcal{C}\) on vertex set \([n]\) with \(n\geq N_d\)
contains a set \(X\subseteq V(T)\) of order less than \(n^{1-\epsilon_d}\) such that
every vertex of \(T\) has an \(n^{1-\epsilon_d}\)-near twin in \(X\).
Let \(f\colon V(T)\to X\) be a function that assigns to each vertex such an \(n^{1-\epsilon_d}\)-near twin.

Then the tournament \(T\) is fully determined by
\begin{enumerate}
\item the choice of the set \(X\subseteq [n]\),
\item the tournament \(T[X]\),
\item the function \(f\colon [n]\to X\),
\item for each vertex \(v\in V(T)\) the set \(N^+(v)\symdiff N^+(f(v))\) of size less than \(n^{1-\epsilon_d}\).
\end{enumerate}
Since the number of subsets of \([n]\) of size at most \(n^{1-\epsilon_d}\) is bounded by \(n^{n^{1-\epsilon_d}}\),
the number of different choices for the above data is at most
\[n^{n^{1-\epsilon_d}}\cdot
  2^{\binom{n^{1-\epsilon_d}}{2}}\cdot
  \left(n^{1-\epsilon_d}\right)^n\cdot
  \left(n^{n^{1-\epsilon_d}}\right)^n
=2^{O(n^{2-\epsilon_d}\log n)}.\]
By choosing \(\epsilon\in(0,\epsilon_d)\), this is bounded by
\(O(2^{n^{2-\epsilon}})\) and thus, there are only \(O(2^{n^{2-\epsilon}})\) labeled tournaments of order \(n\) in the class~\(\mathcal{C}\).
\end{proof} 
\end{document}